\documentclass[journal, onecolumn]{IEEEtran}

\usepackage{amsthm,amsmath,amssymb,bm,cite,algorithm,algpseudocode,float,color,tikz}
\usetikzlibrary{positioning,arrows.meta,fit}
\usepackage{optidef}
\usepackage{booktabs}
\usepackage{bbm}

\newtheorem{theorem}{Theorem}
\newtheorem{lemma}{Lemma}
\newtheorem{remark}{Remark}

\title{Age of Incorrect Information for Pull-Based State Estimation of General Markov Sources}
\author{Marco Zanni, Mohamad Assaad, and Touraj Soleymani
\thanks{
M.~Zanni and M.~Assaad are with the CentraleSup\'{e}lec, University of Paris-Saclay, France (e-mails: {\tt\footnotesize marco.zanni@centralesupelec.fr} and {\tt\footnotesize mohamad.assaad@centralesupelec.fr}). T.~Soleymani is with the University of London, United Kingdom (e-mail: {\tt\footnotesize touraj.soleymani@citystgeorges.ac.uk}). }%
}

\begin{document}
\bstctlcite{BSTcontrol}

\maketitle

\begin{abstract}
We study pull-based remote state estimation of an arbitrary, multi-state Markov source while accounting for both freshness and correctness attributes of information. To that end, we formulate a discounted optimization problem in terms of the age of incorrect information (AoII), and express it as a joint source-AoII belief Markov decision process (MDP) under maximum a posteriori (MAP) estimation. We then exploit the information structure of the model and prove that every reachable belief is represented by the last successfully observed source state and the number of time slots elapsed since that observation. For numerical computation, we truncate the elapsed no-success duration at a finite level and derive an explicit error bound and a criterion for selecting the truncation parameter. For reliable links, we show that an optimal policy can be represented by a look-up table of waiting times. For unreliable links, we propose a persistent policy and derive computable performance bounds. We also show that the MAP estimate stabilizes after a finite number of time slots. To further reduce memory requirements, we introduce a hybrid estimator with an early stationary switch and derive a computable bound on the resulting difference in performance. Finally, we extend the framework to multiple sources, formulate the scheduling problem as a restless multi-armed bandit, establish a sufficient condition for indexability, and develop an approximate Whittle index policy based on interpolation. Our numerical results illustrate the structure of the optimal single-source policy, evaluate the performance of the multi-source policies, and verify that the proposed heuristic policies closely approach the optimal solution while substantially reducing computational efforts.
\end{abstract}


\section{Introduction}\label{sec:introduction}
Remote state estimation is a fundamental component of networked cyber-physical systems, such as those in industrial automation, intelligent transportation, and smart grids. In these applications, a sensor observes a stochastic process and communicates status updates to a remote monitor, which should approximately reconstruct the current state based on the information communicated over a possibly imperfect link. Note that, as the exchanged information needs to be used in a downstream real-time decision-making task related to physical systems, the remote state estimation performance must account for not only correctness but also freshness. The age of information (AoI) is a common metric for quantifying information freshness in status-update systems \cite{kaul2012realtime, yates2019age, yates2021survey}. Let $u_t$ be the generation time of the freshest update that has been successfully delivered to the receiver by time $t$. The AoI is defined as the time elapsed since the generation of the latest available update, i.e.,
\begin{equation}
    A_t = t-u_t.
    \label{eq:intro_aoi}
\end{equation}
The AoI has been extensively studied in queueing, sampling, scheduling, remote-estimation, and integrated sensing and communication systems \cite{sun2017update, sun2020sampling, ornee2021sampling, zanni2026aoiisac}. However, this metric does not directly account for the dynamics of the monitored process or for the correctness of the estimate. An old update may remain fully accurate if the source has not changed, whereas a recently generated update may already be incorrect after a source transition. This limitation has motivated a broad family of freshness metrics, including query AoI, age penalties, binary freshness, uncertainty of information, and estimation costs \cite{chiariotti2022query, champati2022detecting, akar2024query, chen2022uoi, chen2024indexuoi}. The age of incorrect information (AoII) was introduced in \cite{maatouk2020age} to combine information freshness with estimation correctness. Let $X_t$ be the source state and let $\hat X_t$ be the estimate available at the monitor. Assuming that the source and the monitor have been synchronized at least once, the AoII can be written as
\begin{equation}
    \Delta_t
    = t-\max\big\{\tau\in\{0,\ldots,t\}:X_\tau=\hat X_\tau\big\}.
    \label{eq:intro_aoii}
\end{equation}
Equivalently, $\Delta_t$ is zero whenever $X_t=\hat X_t$ and increases by one for every consecutive time slot in which the estimate remains incorrect. Unlike the AoI, the AoII may therefore return to zero without a new transmission if the source evolves back to the state currently estimated by the monitor. This dependence on both the source trajectory and the duration of the estimation error makes the AoII a more comprehensive metric for cyber-physical systems \cite{uysal2022semantic, saha2022relationship, luo2025consecutive}.

\subsection{Different Architectures and Related Work}
Status-updating problems based on the AoII have been studied primarily for Markov sources. For discrete-time sources, constrained Markov decision process (MDP) formulations have been used to characterize the optimal policy structure. In particular, it has been shown that for binary, symmetric, or special Markov sources, the optimal policy is threshold \cite{maatouk2020age, maatouk2023enabler, kriouile2021tracking, kriouile2023pull}. Subsequent works have incorporated unreliable channels, random transmission delays, retransmissions, and hybrid automatic repeat request mechanisms \cite{chen2024randomdelay, bountrogiannis2025harq}. General functions of the AoII have been considered in push remote estimation through threshold policies, see \cite{cosandal2025functions}. These studies show that the optimal update rule generally depends on the source dynamics, the communication cost or constraint, and the duration of the current mismatch. Nevertheless, for general asymmetric sources, a single AoII threshold need not be optimal. A substantial part of the existing literature (see \cite{maatouk2020age,maatouk2023enabler,chen2024randomdelay,zakeri2025pomdp,cocco2023remote}) relies on binary sources or on symmetric Markov processes. In contrast, we consider a general Markov source with an arbitrary transition matrix, without imposing binary or symmetry assumptions.

A parallel line of work considers continuous-time Markov sources. Analytical models based on synchronization cycles and absorbing Markov chains have been developed for push-pull sampling \cite{cosandal2024modeling,akar2024query}. More recently, the policy minimizing the AoII for a general continuous-time Markov source has been shown to have a threshold structure, with thresholds depending on both the source and estimation states, and the corresponding constrained problem has been formulated as a semi-Markov decision process \cite{cosandal2025multithreshold}. Related remote state estimation studies have investigated detection of state transitions, state reconstructions, nonlinear costs for persistent errors, and scheduling under communication constraints \cite{champati2022detecting, luo2025consecutive}.

Most of the above AoII formulations assume that the decision maker knows the current source state, the current estimation error, or the instantaneous AoII. In particular, the AoII is often treated as a state variable directly observed by the scheduler. Such information is naturally available in a push-based architecture, where the sensor observes the source and can decide whether to transmit, but it is generally unavailable in a pull-based architecture, where the monitor must decide whether to request an update before observing the source. In a pull-based architecture, however, the update decision is made under partial information, and the monitor must infer the source and the estimation cost from its past observations. Partially-observable Markov decision processes (POMDPs) provide the standard framework for such problems: the observation history is summarized by a belief state, which evolves as a controlled Markov process \cite{smallwood1973optimal, kaelbling1998planning}. Such formulations have consequently appeared in partially-observed status update problems \cite{shao2022partially, zakeri2025pomdp, cocco2023remote}. 

Pull architectures have also been studied from a task-oriented perspective. In \cite{talli2025pragmatic}, communication and control policies are jointly designed for a Markov process over a costly channel. However, the objective is the trade-off between control reward and communication cost, and does not involve age metrics. \cite{liyanaarachchi2026age} optimizes job submission decisions for a Markov machine using the state and age of the available estimate, while the sampling process is given exogenously. These formulations differ from the present problem, where the pull decision itself is optimized to control the freshness and correctness of the estimate.

In the present paper, we study pull-based remote state estimation of a Markov source under the AoII metric. The monitor observes the source only through possibly unsuccessful transmissions. A work that is closely related to our study is \cite{cosandal2025joint}, which considers pull-based remote state estimation of a general discrete-time Markov source when the monitor cannot observe the instantaneous AoII, introduces the joint conditional distribution of the source state and AoII as a sufficient statistic, employs a maximum a posteriori (MAP) estimation, and writes the corresponding belief-MDP optimality equations. The formulation in \cite{cosandal2025joint} is the starting point of the present paper, but our objective and treatment are different. While \cite{cosandal2025joint} establishes belief-MDP formulation, and provides heuristic method based on reinforcement learning. Our goal is to turn this partially-observed model into an analytically and computationally tractable framework using finite approximations, developing novel efficient policies, and providing rigorous analysis of performance bounds. Unlike \cite{cosandal2025joint}, we formulate the MDP directly on the set of beliefs reachable from the synchronized initial condition, rather than on the entire probability simplex, and we modify the timing convention within a time slot. We propose multiple heuristic solutions to the problem, providing bounds on their suboptimality. A related multi-sensor formulation is considered in \cite{cosandal2025sensor}, where a successful update reveals only a partial state, and the resulting belief remains a continuous variable; in contrast, in our model, a successful update reveals the complete source state, and we obtain an exact discrete reduction of the MDP.

Recent works have also investigated the structure of estimators in pull architectures, see \cite{liyanaarachchi2026multistage,liyanaarachchi2026beyond}. These works study the approximation of the MAP estimator through piecewise constant functions. Our setting differs in that we consider a different kind of estimator to approximate the MAP.

A generic POMDP is commonly described over the whole probability simplex, which is uncountable even when the underlying hidden state space is finite. In our setting, instead, the monitor is initialized in a synchronized state, and each action produces only a finite number of possible observations. Consequently, only finitely many beliefs can be generated, and the union over all depths is a countable set of states. We reformulate an MDP directly on this set of states rather than on the complete continuous probability simplex. This distinction is essential: the finite computational model is built by exploring beliefs that can actually occur under admissible observation histories, without imposing an arbitrary discretization over the simplex. The use of MAP estimation is also important in this setting. Retaining the last received source state as the estimate is convenient, but it ignores the evolution of a Markov source between updates. Under a MAP estimator, the estimate is instead determined by the current belief and may change even in the absence of a new packet. Recent work (see \cite{zakeri2025pomdp}) has similarly shown that, for Markov sources, the age and content of the latest observation jointly determine the usefulness of outdated information and the structure of the optimal policy. In our problem, this dependence is exactly the reason why a source-state marginal distribution alone is not sufficient: the monitor must keep the joint distribution of the source state and the elapsed duration of incorrect estimation.

\subsection{Main Contributions and Organization}
In this paper, we study pull-based remote state estimation where a remote monitor tracks a general Markov source over a possibly unreliable link. The underlying information pattern is realistic but challenging: at the time of each request, the monitor observes neither the true source state nor the instantaneous AoII. Moreover, a transmitted update may fail. Therefore, unlike age problems in which the current age is directly available to the decision maker, the quantity to be minimized must itself be inferred from the observation history through a state-AoII belief. Our objective is to exploit the information structure, obtain an exact reduced representation, and develop computationally tractable and simplified scheduling rules with explicit performance guarantees. Our main contributions are summarized as follows.

\begin{itemize}
    \item We derive the exact recursion for the joint source-AoII belief and formulate the decision-making problem as a belief-MDP on the set of beliefs reachable from the synchronized initial condition. We then prove that every reachable prior belief is determined by two observable quantities: the last successfully observed source state and the elapsed time since that observation.

    \item We develop a finite approximation by truncating only the elapsed no-success duration at a level $H$. The exact costs and transitions are preserved inside the truncated region, while a null transition from the boundary is assigned zero continuation value. We derive an explicit upper bound on the resulting approximation error and a direct rule for selecting $H$ for a prescribed tolerance. 
    
    \item We exploit the regenerative structure created by successful transmissions to obtain efficient policies and measure their performance certificates. For a reliable link, we show that an optimal look-up-table policy described by waiting times. For an unreliable link, we introduce a persistent policy and derive rigorously computable regenerative upper and lower bounds that certify its suboptimality gap. We then introduce a hybrid estimator with an early stationary switch that considers the stationary distribution after a prescribed time and derive a computable bound on the resulting difference in optimal value.

    \item We extend the reduced-state framework to multiple independent sources sharing a limited number of transmission resources. The scheduling problem is formulated as a restless multi-armed bandit, following the index methodology used in standard literature, (see, e.g., \cite{whittle1988restless, chen2022uoi, chen2024indexuoi, zanni2026aoiisac}). We establish a sufficient condition for indexability, describe the computation of Whittle indices on the finite reduced state spaces, and propose an approximate Whittle index policy based on interpolation from a small collection of anchor states.

    \item We finally provide extensive numerical results that illustrate the structure of the optimal policy over the reduced state space in the single-source setting, compare the persistent policy with the optimal solution, and evaluate the exact and approximate index policies in the multi-source setting. The results show that the proposed efficient policies can closely approach their optimal counterparts while substantially reducing offline computation and online storage requirements.
\end{itemize}
The rest of the paper is organized as follows. The networked system model and the main optimization problem are introduced in Sections \ref{sec:system_model}--\ref{sec:problem_formulation}. Section \ref{sec:model_reduction} establishes the exact reduced representation, develops the finite truncation of the state space, and derives its error bound. Section \ref{sec:reliable_link} studies the reliable link case and develops a look-up-table policy, while Section \ref{sec:unreliable_persistent_policy} tackles the unreliable link case and develops the persistent policy and its performance certificate. Section \ref{sec:hybrid_estimator} proves MAP stabilization and develops the hybrid estimator. Section \ref{sec:multi_source} addresses the multi-source formulation and the index policies. Section \ref{sec:numerical_results} presents numerical results. Finally, Section \ref{sec:conclusion} concludes the paper.

\section{Networked System Model}\label{sec:system_model}
In this section, we introduce the Markov source, the pull-based communication model, and the within-slot sequence of estimation and source evolution. We then define the AoII dynamics and construct the joint source-AoII belief and its MAP-based update equations. This model specifies the information available to the monitor when each pull decision is made.

\subsection{Source and Channel Models}
In our system, time is divided into slots indexed by $t\in\mathbb Z_{\geq 0}$. A sensor observes a Markov process $\{X_t\}$ with state space $\mathcal X = \{1,\dots,N\}$ and transition matrix $\bm P$. The process remains in state $X_t$ during time slot $t$ and transitions to the next state at the end of the slot. At the beginning of time slot $t$, the monitor chooses an action $a_t\in\{0,1\}$. The action $a_t=1$ means that the monitor sends a pull request, while $a_t=0$ means that no request is sent. When a request is sent, it reaches the source instantaneously, and the sensor samples the current state $X_t$ and attempts to deliver it to the monitor. The transmission is successful with probability $s\in[0,1]$, independently of the source evolution. A successfully transmitted packet becomes available to the monitor within the same time slot. 

The observation available to the monitor at time slot $t$ is denoted by $o_t\in\mathcal O(a_t)$, where
$\mathcal O(0)=\{\varnothing\}$ and
$\mathcal O(1)=\{\varnothing,1,2,\dots,N\}$. More precisely,
\begin{align}
    o_t=
    \begin{cases}
        X_t, & a_t=1 \text{ and transmission is successful},\\
        \varnothing, & \text{otherwise}.
    \end{cases}
    \label{eq:obs}
\end{align}
Thus, the null observation $\varnothing$ represents either an idle slot or an unsuccessful transmission attempt.

\begin{figure}[t!]
\centering
\begin{tikzpicture}[
    font=\footnotesize,
    >=stealth,
    proc/.style={
        draw,
        rounded corners,
        thick,
        align=center,
        minimum width=1.65cm,
        minimum height=0.85cm
    },
    chan/.style={
        draw,
        thick,
        align=center,
        minimum width=1.35cm,
        minimum height=0.6cm
    },
    arr/.style={->, thick}
]

\node[proc] (mon) at (0,0) {Monitor\\$\hat X_t$};
\node[proc] (src) at (5.4,0) {Source\\$X_t$};
\node[chan] (ch) at (2.7,-1.05) {channel\\$s$};

\draw[arr] (mon.north east) .. controls (1.4,0.85) and (4.0,0.85) ..
    node[above] {$a_t=1$} (src.north west);

\draw[arr] (src.south west) -- node[above,sloped] {$X_t$} (ch.east);
\draw[arr] (ch.west) -- node[above,sloped] {$o_t$} (mon.south east);

\node[align=center] at (2.7,-2.05)
{\scriptsize $b_t \;\rightarrow\; a_t \;\rightarrow\; o_t \;\rightarrow\; \hat X_t^{+},\,\Delta_t^{+} \;\rightarrow\; P \;\rightarrow\; b_{t+1}$};

\end{tikzpicture}
\caption{Information flow and update order within one decision slot.}
\label{fig:sys}
\end{figure}
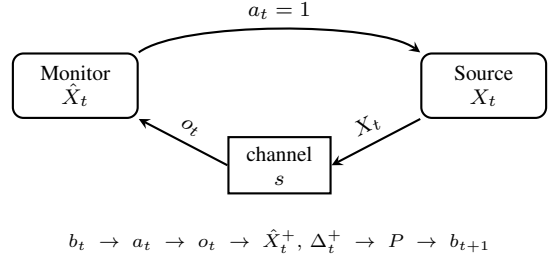

\subsection{Age of Incorrect Information}
The AoII component takes values in $\mathbb Z_{\geq0}$. Before the observation for time slot $t$ is received, the monitor uses the estimate $\hat X_t^-=\hat{x}(b_t)$. Once $o_t$ is available, the estimate for the current slot becomes
\begin{align}
\hat X_t^+ =
\begin{cases}
    k, & o_t=k,\ k\in\{1,\dots,N\},\\
    \hat{x}(b_t), & o_t=\varnothing.
\end{cases}
\label{eq:post_obs_estimator}
\end{align}
To make the within-slot timing explicit, we distinguish the AoII immediately before and after the current observation. Let $\Delta_t^-$ be the prior AoII (immediately before the observation) for time slot $t$, and let $\Delta_t^+$ be the posterior AoII (immediately after the observation and the corresponding estimate update). Since the cost of time slot $t$ is evaluated after this update, its instantaneous AoII contribution is $\Delta_t^+$. The posterior AoII satisfies
\begin{align}
\Delta_t^+ =
\begin{cases}
    0, & X_t=\hat X_t^+,\\
    \Delta_t^-, & X_t\neq \hat X_t^+.
\end{cases}
\label{eq:post_obs_aoii}
\end{align}
After the source transition, the AoII at the beginning of the next time slot evolves as
\begin{align}
\Delta_{t+1}^- =
\begin{cases}
    0, & X_{t+1}=\hat X_{t+1}^-,\\
    \Delta_t^+ +1, & X_{t+1}\neq \hat X_{t+1}^-.
\end{cases}
\label{eq:aoii}
\end{align}
Therefore, synchronization between the source and the monitor resets the AoII to zero; otherwise, the AoII increases by one.

\subsection{Joint Beliefs and Maximum-A-Posteriori Estimates}
The sequence of operations within a time slot is as follows. At each time, the monitor holds the prior and posterior beliefs $b_t$ and $\tilde{b}_t$ about the source state and the AoII. Before choosing $a_t$ and receiving $o_t$, the monitor holds the prior belief $b_t$. The action is then selected and the corresponding observation is revealed. Based on this observation, the monitor updates its posterior belief. The source subsequently evolves according to $\bm P$, yielding the prior belief $b_{t+1}$ at the beginning of the next time slot. The controlled process starts immediately after a synchronization event.

For a belief $b$ over $\{1,\dots,N\}\times\mathbb Z_{\geq0}$, define
\begin{align}
\pi_b(i)=\sum_{\Delta=0}^{\infty} b(i,\Delta), \qquad
\hat{x}(b)=\arg\max_{i\in\{1,\dots,N\}} \pi_b(i),
\label{eq:map_belief}
\end{align}
where $\pi_b(i)$ indicates the source-state marginal distribution associated with $b$. Since the source state is revealed only after a successful transmission, the monitor generally cannot observe either $X_t$ or $\Delta_t^-$. It therefore maintains their joint conditional distribution,
\begin{align}
b_t(i,\Delta)=\mathbb P(X_t=i,\Delta_t^-=\Delta\mid H_t),
\label{eq:b}
\end{align}
where $H_t$ collects all observations and actions available before $a_t$ is chosen. Before time slot $0$, a successful update has revealed the current source state to the monitor. Therefore, for $i_0\in\{1,\dots,N\}$, the initial prior belief is
\begin{align}
b_0(i,\Delta)=\mathbbm{1}\{i=i_0,\Delta=0\}.
\label{eq:synchronized_initial_belief}
\end{align}
Equivalently, $X_0=\hat X_0^-=i_0$ and $\Delta_0^-=0$. This assumption identifies the class of beliefs that can be reached from the initial condition and may be interpreted as an initial synchronization step.

For a prior belief $b_t$ and a realized observation $o_t$, denote by $\tilde b_t$ the corresponding posterior belief:
\[
\tilde b_t(i,\Delta)
=
\mathbb P(X_t=i,\Delta_t^+=\Delta\mid H_t,a_t,o_t).
\]
For convenience, define the reset map
\begin{align}
r_{\rm reset}(i,\delta,\hat x)=
\begin{cases}
    0, & i=\hat x,\\
    \delta, & i\neq \hat x.
\end{cases}
\label{eq:R}
\end{align}
Using \eqref{eq:obs} and \eqref{eq:post_obs_aoii}, for any $k\in\{1,\dots,N\}$, the posterior belief is determined by the map $U$:
\begin{align}
&\tilde b_t(i,\Delta) =U(b_t,o_t)(i,\Delta) \notag\\
&=\begin{cases}
\displaystyle \sum_{\delta=0}^{\infty} b_t(i,\delta) \mathbbm{1}\!\left\{ \Delta={}r_{\rm reset}(i,\delta,\hat x(b_t))\right\}, & o_t=\varnothing,\\[4\jot]
\mathbbm{1}\{i=k,\Delta=0\}, & o_t=k.
\end{cases}
\label{eq:hat_b_t}
\end{align}
When $o_t=\varnothing$, the monitor receives no new information about $X_t$ and keeps the estimate $\hat x(b_t)$. The AoII is therefore reset only for source states equal to that estimate. For beliefs reachable under \eqref{eq:aoii}, this operation satisfies
$U(b_t,\varnothing)=b_t$, since any mass associated with
$X_t=\hat x(b_t)$ already has zero AoII. If instead $o_t=k$, the monitor learns that $X_t=k$, updates its estimate to $k$, and assigns zero posterior AoII.

After this update, the source undergoes a transition according to $\bm P$. Given $\tilde b_t$, define the MAP estimate for the next time slot as
\begin{align}
\hat{x}_{\rm next}(\tilde b_t) =
\arg\max_{i\in\{1,\dots,N\}} \sum_{m=1}^{N}\sum_{\delta=0}^{\infty} \tilde b_t(m,\delta)P_{mi}.
\label{eq:pred_map}
\end{align}
Thus,
$\hat X_{t+1}^-=\hat{x}_{\rm next}(\tilde b_t)$.
The next prior belief is determined by the map $F$:
\begin{align}
&b_{t+1}(i,\Delta)=F(b_t,o_t)(i,\Delta) \notag\\
&=\sum_{m=1}^{N}\sum_{\delta=0}^{\infty}
\tilde b_t(m,\delta)P_{mi} \mathbbm{1}\!\left\{
\Delta=G(i,\delta, \hat{x}_{\rm next}(\tilde b_t))
\right\},
\label{eq:b_t}
\end{align}
where
\begin{align}
G(i,\delta,\hat x) =
\begin{cases}
    0, & i=\hat x,\\
    \delta+1, & i\neq \hat x.
\end{cases}
\label{eq:G}
\end{align}
Hence, \eqref{eq:b_t} first propagates the source distribution through $\bm P$, then determines the MAP estimate. Probability mass assigned to the estimated state is placed at zero AoII, whereas the AoII associated with all other states is incremented by one. Throughout the paper, ties in the MAP rules are resolved according to a fixed deterministic ordering.

For later use, let $\rho(o|b,a)$ be the probability of obtaining observation $o$ when action $a$ is selected at belief $b$. Under the channel model,
\begin{align}
\rho(k|b,1)&=s\pi_b(k), \qquad k\in\{1,\dots,N\}, \label{eq:rho_k}\\[1\jot]
\rho(\varnothing|b,1)&=1-s, \label{eq:rho_empty_1}\\[1\jot]
\rho(\varnothing|b,0)&=1, \label{eq:rho_empty_0}
\end{align}
with all remaining observation probabilities equal to zero. Since the next belief is fully determined by the current belief and the observation, $b_t$ summarises all relevant information contained in the history:
\begin{align}
\mathbb P(b_{t+1}|H_t,a_t,o_t)
=
\mathbb P(b_{t+1}|b_t, a_t, o_t).
\end{align}
The resulting transition probabilities between belief states are therefore
\begin{align}
T(b'|b,a)
&=
\mathbb P(b_{t+1}=b'|b_t=b,a_t=a)\\
&=
\sum_{o\in\mathcal O(a)}
\rho(o|b,a)\mathbbm{1}\{b'=F(b,o)\}.
\label{eq:Tb}
\end{align}

\section{Remote State Estimation Problem} \label{sec:problem_formulation}
In this section, we formulate the discounted tradeoff between AoII and pull-request costs. Using the joint source-AoII belief as the information state, we define the corresponding belief-MDP and derive its Bellman optimality equations. This formulation serves as the reference problem for the reduced-state analysis developed subsequently.

\subsection{Main Optimization Problem}
Since the AoII cost of time slot $t$ is evaluated after the observation, the instantaneous AoII cost is $\Delta_t^+$. For a history dependent admissible policy $\phi$, we define the discounted AoII cost as
\begin{align}
J_\gamma^\phi(b_0) =
\mathbb E_{b_0}^{\phi} \Big[ \sum_{t=0}^{\infty}\gamma^t\Delta_t^+ \Big],
\label{eq:discounted_aoii_cost}
\end{align}
where $\gamma\in(0,1)$ is the discount factor. Similarly, we define the discounted sampling usage as
\begin{align}
S_\gamma^\phi(b_0) =
\mathbb E_{b_0}^{\phi} \Big[ \sum_{t=0}^{\infty}\gamma^t a_t \Big].
\label{eq:discounted_sampling_cost}
\end{align}
We assign a cost $\lambda \geq 0$ to each pull request. The resulting penalized discounted optimization problem is stated as
\begin{align}
\underset{\phi}{\text{minimize}} \Big\{J_\gamma^\phi(b_0)+\lambda S_\gamma^\phi(b_0)\Big\},
    \label{Opt1}
\end{align}
where parameter $\lambda$ controls the tradeoff between the AoII and the transmission cost. The associated optimal value of the problem is
\begin{align}
V^\lambda(b_0) =
\inf_\phi \mathbb E_{b_0}^{\phi} \Big[
    \sum_{t=0}^{\infty}\gamma^t \left( \Delta_t^+ + \lambda a_t \right)\Big].
\label{eq:discounted_lagrangian_problem}
\end{align}

\subsection{Equivalent Markov Decision Process}
For a fixed $\lambda$, the problem yields a discounted belief-MDP. By varying $\lambda$, one obtains policies corresponding to different AoII-transmission tradeoffs. Let $s_t=(X_t,\Delta_t^-)\in\mathcal S$ denote the hidden physical state before the observation, where $\mathcal S=\{1,\dots,N\}\times\mathbb Z_{\geq0}$. This state is not available to the monitor. Moreover, because of the MAP estimation, the transition law depends on the estimate, which is a deterministic function of the belief; hence, $(X_t,\Delta_t^-)$ alone is not a sufficient controlled Markov state. On the other hand, the update rule in \eqref{eq:b_t} induces a belief-MDP. Differently from the standard continuous belief simplex, here we explicitly restrict the state space to the belief states that are reachable from the initial belief.

Let $\mathcal B_0=\{b_0\}$ be the set containing only the initial belief. Recursively, for every $t\geq0$, define
\begin{align}
\mathcal B_{t+1} \!= \!\Big\{ F(b,o)\!:\! b\!\in\!\mathcal B_t, a\!\in\!\mathcal A, o\!\in\!\mathcal O(a), \rho(o|b,a)\!>\!0 \Big\}.
\label{eq:reachable_Bt}
\end{align}
The overall reachable belief space is then given by
\begin{align}
\mathcal B=\bigcup_{t\geq0}\mathcal B_t.
\label{eq:reachable_B}
\end{align}
At every finite depth, only finitely many beliefs are reachable. Hence, the overall reachable belief space is at most countable. For a fixed transmission cost $\lambda$, the discounted belief-MDP is defined as $(\mathcal B,\mathcal A,T,c^\lambda,\gamma)$, where
\begin{itemize}
    \item The state space is the reachable belief space $\mathcal B$ in \eqref{eq:reachable_B}.
    \item The action space is $\mathcal A=\{0,1\}$, and each action $a\in\mathcal A$ is feasible for every $b\in\mathcal B$.
    \item The transition probability $T(b'|b,a)$ is induced by the deterministic successor map $F$ and by the observation probabilities $\rho(o|b,a)$, as in \eqref{eq:Tb}.
    \item The cost function is the expected posterior AoII plus the sampling cost:
    \begin{equation}
    c^\lambda(b,a) = (1-as) \sum_{i=1}^{N} \sum_{\Delta=0}^{\infty} \Delta b(i,\Delta) + \lambda a.
    \label{eq:cost_function}
    \end{equation}
    \item $\gamma\in(0,1)$ is the discount factor.
\end{itemize}

Finally, the Bellman equation for the discounted penalized problem can be written on the reachable belief space as
\begin{align}
V^\lambda(b) &= \min_{a\in\mathcal A} \Big \{ c^\lambda(b,a) \notag\\
&\qquad \qquad + \gamma\sum_{o\in\mathcal O(a)}\rho(o|b,a) V^\lambda\!\left(F(b,o)\right)\Big\},
\label{eq:bell}
\end{align}
for $b\in\mathcal B$. Equivalently, by defining the action-value function,
\begin{align}
Q^\lambda(b,a) &= c^\lambda(b,a) + \gamma\sum_{o\in\mathcal O(a)}\rho(o|b,a)\nonumber\\
&\qquad \qquad \qquad \quad \times \min_{a'\in\mathcal A} Q^\lambda\!\left(F(b,o),a'\right),
\label{eq:q_bell}
\end{align}
the optimal action at belief $b$ is obtained as $a^*(b)\in\arg\min_{a\in\mathcal A}Q^\lambda(b,a)$.

The belief formulation above is exact, but it does not yet expose the particular structure of the beliefs that can actually occur. In the next section, we show that, under transmission failures and perfect state revelation upon a successful update, the belief is determined by only two observable quantities: the last observed source state and the elapsed time since that observation. This reduction is the basis of all subsequent analysis and numerical computation.

\section{Model Reduction Analysis}\label{sec:model_reduction}
The belief-MDP is defined on a set of probability distributions. This set has a much stronger structure than a generic POMDP belief space. A successful transmission reveals the current source state exactly and resets the AoII to zero. Between two successful transmissions, every observation is null. Since the channel success event is independent of the source state, a failed request and an idle time slot provide no information about the source. Consequently, after the latest successful update, the belief evolves deterministically until the next success. This property lets us reduce the original model to a simpler representation without loss of information. let the timeline be subdivided into no-success cycles, each one beginning immediately after a successful transmission and finishing at the following successful transmission. We formalize this reduction. We first show that every reachable belief is represented by the last observed state and the elapsed no-success duration. We then obtain an equivalent reduced-state MDP which contains the same information as the original one and allows further analyses. Its only unbounded coordinate is the elapsed duration, which leads to a finite approximation controlled by a single parameter. Finally, we derive an explicit bound on the approximation error and describe the finite model used for computation.

\subsection{Parameterization of Reachable Beliefs}
We first introduce the quantities associated with a no-success cycle. Let $e_i$ be the $i$-th canonical row vector and define the synchronized belief
\begin{equation}
b_i(x,\Delta)=\mathbbm{1}\{x=i,\Delta=0\}.
\label{eq:degenerate_synchronized_belief}
\end{equation}
Immediately after a successful transmission revealing state $i$, the belief at the first subsequent decision slot is
\begin{equation}
b_{i,1}:=F(b^i,i),
\label{eq:first_post_success_belief}
\end{equation}
and, as long as no further successful transmission is received, it evolves according to
\begin{equation}
b_{i,n+1}=F(b_{i,n},\varnothing),\qquad n\geq1.
\label{eq:belief_recursion}
\end{equation}
For every $i\in\mathcal X$ and $n\geq1$, define the source-state marginal distribution
\begin{equation}
p_i(n):=e_i\bm P^n,
\label{eq:reduced_source_marginal}
\end{equation}
the corresponding MAP estimate
\begin{equation}
\hat x_i(n)\in\arg\max_{j\in\mathcal X}(e_i\bm P^n)_j,
\label{eq:reduced_map_estimate}
\end{equation}
and the diagonal mismatch matrix
\begin{equation}
\bm D_i(n) :=
\operatorname{diag} \! \big(\mathbbm{1}\{1\neq \hat x_i(n)\}, \dots, \mathbbm{1}\{N\neq \hat x_i(n)\}\big).
\label{eq:reduced_mismatch_matrix}
\end{equation}
The expected posterior AoII in a no-success slot is
\[
g_i(n):=\sum_{x=1}^{N}\sum_{\Delta=0}^{\infty}\Delta b_{i,n}(x,\Delta),
\]
which admits the equivalent expression
\begin{equation}
g_i(n) =
\sum_{r=1}^{n}e_i\bm P^r\bm D_i(r)\bm P\bm D_i(r+1) \cdots \bm P\bm D_i(n)\mathbf 1,
\label{eq:reliable_g_closed_form}
\end{equation}
where the product after $\bm D_i(n)$ is omitted when $r=n$.

\begin{theorem}
\label{thm:exact_reduction}
Every prior belief reachable from the synchronized initial condition is either $b_0$ or is of the form $b_{i,n}$ for some $(i,n)\in\mathcal X\times\mathbb N_{\geq1}$. Consequently, the observable pair $(i,n)$ is a sufficient observable state in the belief-MDP problem. Moreover, we have
\begin{equation}
 0\leq g_i(n)\leq n.
 \label{eq:g_linear_bound}
\end{equation}
\end{theorem}

\begin{proof}
See Appendix \ref{app:exact_reduction}.
\end{proof}

By Theorem \ref{thm:exact_reduction}, after the initial time, we can define a fully-observed MDP that is exactly equivalent to the belief-MDP problem. The reduced-state MDP is described by $(\mathcal Y,\mathcal A,T_{\rm red},c^\lambda,\gamma)$, where
\begin{itemize}
    \item The state space is $\mathcal Y = \mathcal X\times\mathbb N_{\geq1}$.
    The state $(i,n)$ records the latest successfully observed source state and the elapsed no-success duration.
    \item The action space is $\mathcal A=\{0,1\}$, and each action is feasible at every state $(i,n)\in\mathcal Y$.
    \item The transition probabilities are
    \begin{align}
    T\bigl((i,n+1)\mid(i,n),a\bigr)
    &=1-sa,
    \label{eq:reduced_null_transition}\\
    T\bigl((k,1)\mid(i,n),a\bigr)
    &=sa(\bm P^n)_{ik},\qquad k\in\mathcal X.
    \label{eq:reduced_success_transition}
    \end{align}
    \item The cost function is
    \begin{equation}
    c^\lambda(i,n,a)=(1-sa)g_i(n)+\lambda a.
    \label{eq:reduced_stage_cost}
    \end{equation}
    \item $\gamma\in[0,1)$ is the discount factor.
\end{itemize}

Defining $V_i(n)$ as the optimal value at reduced state $(i,n)$, the corresponding Bellman equation is
\begin{align}
V_i(n) & =\min_{a\in\{0,1\}}\Big\{ (1-sa)g_i(n)+\lambda a \notag\\
&+\gamma(1-sa)V_i(n+1) +\gamma sa\sum_{k=1}^{N}(\bm P^n)_{ik}V_k(1) \Big\}. \label{eq:general_reduced_bellman}
\end{align}
For the synchronized initial belief in \eqref{eq:synchronized_initial_belief}, the source state is already known and the AoII is zero. A request cannot improve the current information and only adds the cost $\lambda$. Hence, the passive action is optimal and
\begin{equation}
 V^\lambda(b_0)=\gamma V_{i_0}(1).
 \label{eq:initial_to_reduced_value}
\end{equation}
Note that the monitor does not observe the physical pair $(X_t,\Delta_t)$, but it does know the reduced state $(i,n)$ because both the latest successful observation and the elapsed time since that observation are part of its own information history. The reduction is exact and does not discretize or approximate the joint belief. It also shows that the state set grows only linearly with the largest elapsed duration considered in computation.

\subsection{Truncation of State Space}
As the belief-MDP problem has an unbounded state space, finding the optimal solution would require solving the Bellman equation on infinitely many states. Indeed, the $n$ coordinate of the beliefs can take values in all $\mathbb N_0$. In practical implementations, this is not feasible, and the belief-MDP must be approximated by a finite model, where the $n$ coordinate of the belief is truncated at a value $H$. In this way, the original problem is approximated by a finite dynamic program that works as follows:
\begin{itemize}
    \item Offline, the unbounded AoII dynamics are replaced by their clipped counterpart: after the posterior update and the prediction step, any AoII values larger than $H$ are clipped at the boundary value $H$.
    \item Online, the monitor applies the action prescribed by the dynamic program for the truncated belief.
\end{itemize}
The belief-MDP problem is replaced by a computable finite approximation that preserves the exact evolution inside the explored region and clips only the AoII tail. This finite belief-MDP is defined on the truncated state space 
\begin{equation}
\mathcal Y_H=\mathcal X\times\{1,\dots,H\}.
\label{eq:finite_reduced_state_space}
\end{equation}
The finite-model optimal policy obtained from the truncated model may differ from the optimal policy of the original problem. In this section, we quantify the error induced by this truncation, and we show that the truncated model maintains a controlled approximation of the original problem. For every $H\in\mathbb N$, let $V_H^*(b_0)$ be the optimal value of the finite reduced-state MDP obtained by retaining the states with elapsed no-success duration at most $H$ and assigning zero continuation value after a null transition from the boundary. For sake of simplicity, define $\eta:=1-s$ and $\beta:=\gamma\eta$. Let
\[
q:=\max_{i,j\in\mathcal X}\sum_{\ell\neq j}P_{i\ell}
=1-\min_{i,j\in\mathcal X}P_{ij}.
\]
The constant $q$ is the worst case probability that the source avoids any prescribed estimate, over all current source states and all possible estimates.

For every $n\geq1$, define
\begin{equation}
\overline C_n:=
\begin{cases}
\displaystyle
\frac{\lambda}{1-\beta}
+\frac{\eta q}{1-q}
\left(
\frac{1}{1-\beta}-\frac{q^n}{1-\beta q}
\right),
& 0\leq q<1,\\[3mm]
\displaystyle
\frac{\lambda+\eta n}{1-\beta}
+\frac{\eta\beta}{(1-\beta)^2},
& q=1.
\end{cases}
\label{eq:Cbar_AT_closed_form}
\end{equation}

\begin{theorem}
\label{thm:horizon_truncation}
The error due to the truncation satisfies
\begin{equation}
0\leq V^*(b_0)-V_H^*(b_0) \leq \gamma^{H+1}\overline M_H,
\label{eq:reduced_H_error_bound}
\end{equation}
where
\begin{equation}
\overline M_H :=\overline C_{H+1} + \frac{\gamma s}{1-\gamma}\overline C_1.
\label{eq:M_H_AT}
\end{equation}
\end{theorem}

\begin{proof}
See Appendix \ref{app:horizon_truncation}
\end{proof}

Theorem \ref{thm:horizon_truncation} provides a direct rule for selecting the only approximation parameter. For a desired absolute error $\varepsilon>0$, it is sufficient to choose $H$ such that
\begin{equation}
\gamma^{H+1}\overline M_H\leq\varepsilon.
\label{eq:H_selection_absolute}
\end{equation}
Equivalently, consider the normalized tolerance
\[
\widehat\varepsilon:=(1-\gamma)\varepsilon,
\]
which has the same scale as the cost function. It is sufficient that
\begin{equation}
(1-\gamma)\gamma^{H+1}\overline M_H \leq \widehat\varepsilon.
\label{eq:H_selection_normalized}
\end{equation}
The left-hand side converges to zero as $H\to\infty$: when $q<1$, $\overline M_H$ is uniformly bounded, while for $q=1$ it grows at most linearly in $H$. Thus, the minimum value for $H$ can be obtained through a simple one-dimensional search.

Table \ref{tab:DH_selection_example} illustrates the resulting rule for a representative parameter set. The table reports the smallest value $H$ obtained from the condition \eqref{eq:H_selection_normalized} with $q=0.8$, $s=0.8$, and $\lambda=1$.
\begin{table}[t!]
    \centering
    \caption{Minimum $H$ satisfying \eqref{eq:H_selection_normalized} for $q=0.8$, $s=0.8$, and $\lambda=1$.}
    \label{tab:DH_selection_example}
    \scriptsize
    \begin{tabular}{c|ccccc}
        \toprule
        $\widehat{\varepsilon}$
        & $\gamma=0.50$ & $\gamma=0.70$ & $\gamma=0.85$ & $\gamma=0.90$ & $\gamma=0.95$ \\
        \midrule
        $1$              & $1$  & $1$  & $1$  & $1$  & $3$   \\
        $5\cdot10^{-1}$  & $1$  & $2$  & $5$  & $8$  & $17$  \\
        $10^{-1}$        & $3$  & $7$  & $15$ & $24$ & $48$  \\
        $5\cdot10^{-2}$  & $4$  & $9$  & $20$ & $30$ & $62$  \\
        $10^{-2}$        & $7$  & $13$ & $29$ & $45$ & $93$  \\
        $5\cdot10^{-3}$  & $8$  & $15$ & $34$ & $52$ & $107$ \\
        $10^{-3}$        & $10$ & $20$ & $44$ & $67$ & $138$ \\
        \bottomrule
    \end{tabular}
\end{table}

\subsection{Truncated Markov Decision Process} \label{subsec:belief_space_construction_algorithm}
The model reduction presented makes the offline construction significantly simpler than the generic belief exploration specified in \eqref{eq:reachable_Bt}--\eqref{eq:reachable_B}. For a selected $H$, the truncated state space is the discrete set $\mathcal Y_H$, and the quantities required by the Bellman equation are $\bm P^n$, $\hat x_i(n)$, and $g_i(n)$ for $i\in\mathcal X$ and $1\leq n\leq H$. Algorithm \ref{alg:belief_space_construction} summarizes the offline construction. The algorithm first computes the MAP estimate and expected AoII tables. It then creates the exact transitions inside $\mathcal Y_H$ and sends only the boundary null transitions to the terminal state with zero continuation value. The resulting MDP can be solved by standard value iteration. The number of states (including the terminal state $\partial$) is exactly $NH+1$. More importantly, each state-action pair has only one null successor and at most $N$ successful update successors. The reduced construction therefore preserves the exact information structure of the original belief-MDP while making the source of the numerical approximation explicit: only no-success durations larger than $H$ are omitted.

\begin{algorithm}[t!]
\caption{Truncated Markov Decision Process}
\label{alg:belief_space_construction}
\begin{algorithmic}[1]
\State Choose the no-success truncation level $H$
\State Set $\mathcal Y_H\gets\mathcal X\times\{1,\dots,H\}$ and add the terminal state $\partial$
\For{$n=1,\dots,H$}
    \State Compute $\bm P^n$
    \For{all $i\in\mathcal X$}
        \State Compute $p_i(n)=e_i\bm P^n$ and $\hat x_i(n)$ from \eqref{eq:reduced_map_estimate}
        \State Construct $\bm D_i(n)$ from \eqref{eq:reduced_mismatch_matrix}
    \EndFor
\EndFor
\For{all $(i,n)\in\mathcal Y_H$}
    \State Compute $g_i(n)$ from \eqref{eq:reliable_g_closed_form}
    \For{all $a\in\{0,1\}$}
        \State Set $c^\lambda(i,n,a)\gets(1-sa)g_i(n)+\lambda a$
        \For{all $k\in\mathcal X$}
            \State Set $T_H((k,1)|(i,n),a)\gets sa(\bm P^n)_{ik}$
        \EndFor
        \If{$n<H$}
            \State Set $T_H((i,n+1)|(i,n),a)\gets1-sa$
        \Else
            \State Set $T_H(\partial|(i,H),a)\gets1-sa$
        \EndIf
    \EndFor
\EndFor
\State Assign zero cost and zero continuation value to $\partial$
\State Return the finite MDP on $\mathcal Y_H\cup\{\partial\}$
\end{algorithmic}
\end{algorithm}

\section{Look-Up-Table Policy for Reliable Links}\label{sec:reliable_link}
We now start analyzing the actual solution of our problem in greater depth. We first specialize the reduced-state MDP problem to a reliable link, i.e., $s=1$. Although this constitutes an optimistic special case, we analyze it first since it provides a useful analytical baseline. The structure of the analysis will be useful for the  general case. Every request succeeds, so a transmission always ends the current no-success cycle, reveals the current source state, and regenerates the process at a state $(k,1)$. The next theorem shows that the entire optimal policy is determined by the first time selected for requesting a transmission after each possible observed state.

\begin{theorem}
\label{thm:reliable_waiting_table}
For the reliable-link reduced-state MDP problem, there exists an optimal deterministic policy that is completely described by
\[
\bm m^*=(m_1^*,\dots,m_N^*) \in \left(\mathbb N_{\geq1}\cup\{\infty\}\right)^N.
\]
After observing state $i$, the monitor remains idle for time slots $n=1,\dots,m_i^*-1$, and transmits for the first time at time slot $n=m_i^*$.
\end{theorem}

\begin{proof}
See Appendix \ref{app:reliable_waiting_table}.
\end{proof}

Theorem \ref{thm:reliable_waiting_table} shows that the reliable-link scheduler does not need to store an action for every pair $(i,n)$. For each possible latest observation $i$, it only stores the optimal first request time $m_i^*$. Hence, the table
\[
\begin{array}{c|cccc}
\text{latest observation } i
& 1 & 2 & \cdots & N \\
\hline
\text{first request time } m_i^*
& m_1^* & m_2^* & \cdots & m_N^*
\end{array}
\label{eq:reliable_waiting_table_display}
\]
fully specifies the optimal policy from one successful observation to the next.

\section{Persistent Policy for Unreliable Links}\label{sec:unreliable_persistent_policy}
We now consider the general case with $s\in(0,1)$. The exact reduced representation of Theorem \ref{thm:exact_reduction} remains valid, but a transmission attempt does not necessarily end the current no-success cycle because it may fail. A successful transmission still produces a regeneration: the monitor observes the current source state exactly, the posterior AoII is zero, and the next cycle starts from the new observed state. However, after a failure, the reduced state moves from $(i,n)$ to $(i,n+1)$, and the optimal action may in principle return to idle. Hence, unlike the reliable case, the optimal policy need not be described by a table of waiting times. Our goal here is not to solve the unreliable-link problem in closed form. Instead, we introduce a heuristic persistent policy under which the scheduler always transmits starting from a given time slot, and we provide a theoretical analysis of its performance compared to the optimal policy.

%\subsection{Regenerative Structure under Failed Transmissions}
We organize the system evolution into cycles delimited by successful transmissions. Within each cycle, the latest observed source state remains unchanged, while the elapsed no-success duration increases after every idle slot or failed transmission attempt. This representation separates the deterministic evolution of the monitor’s information from the random occurrence of the next successful update and provides the basis for the regenerative analysis developed here.

Throughout this section, let $\eta:=1-s$ be the failure probability of a transmission attempt. The functions $p_i(n)$, $\hat x_i(n)$, $\bm D_i(n)$, and $g_i(n)$ defined in Section \ref{sec:model_reduction} characterize every no-success cycle. Note that failed attempts do not change the information available to the monitor. Hence, the source marginal distribution is still
\[
p_i(n)=e_i\bm P^n,
\]
the MAP estimate is still $\hat x_i(n)$, and the expected AoII cost in a time slot without a successful update is still $g_i(n)$.

%\subsection{Persistent Policy}

We propose a heuristic persistent policy described by a vector
\[
\bm m=(m_1,\dots,m_N)\in\mathbb N_{\geq1}^N.
\]
After a successful transmission that observes state $i$, the monitor remains idle for time steps $n=1,\dots,m_i-1$, and transmits for the first time at time step $n=m_i$. If this transmission fails, the monitor keeps sending a new pull request until a success occurs. Thus, along the no-success cycle after observing $i$, the action is
\[
a_i(n,\bm m)=\mathbbm{1}\{n\geq m_i\},
\qquad n\geq1.
\]
The vector $\bm m$ is therefore an analogue of the waiting time table of Section \ref{sec:reliable_link}, with the additional persistent request rule after the first attempt. We remark that this rule is proposed as a heuristic structure and it is not claimed to be optimal.

The vector $\bm m$ is a design parameter selected offline; for instance, it can be chosen by minimizing the upper bound $U_{\bm \nu,K}(\bm m)$ derived below over a finite candidate set, or it can be chosen by solving the original MDP, keeping the smallest values for which action $a=1$ is optimal. Once selected, only the $N$ waiting times need to be stored online.

%\subsection{Upper Bound on the Suboptimality of the Persistent Policy}
Using the persistent rule described above allows a substantial simplification during the online implementation of our model. The trade-off is that this policy is suboptimal. Here, we provide an upper bound in closed form on the difference in performance between the persistent policy and the optimal policy. Let $\beta := \gamma\eta = \gamma(1-s)$. For a fixed table $\bm m$, let $V_i(\bm m)$ denote the discounted cost starting from the first time step after a successful observation of state $i$. Fix a finite regenerative horizon $T\geq1$. For
\[
u=(u_1,\dots,u_T)\in[0,1]^T,
\]
where $u_n$ represents the conditional probability of requesting a transmission in time slot $n\in\{1,\dots,T\}$ given that no success has occurred earlier in the cycle, define
\begin{equation}
Q_0(u):=1, \qquad Q_n(u):=(1-su_n)Q_{n-1}(u),
\label{eq:no_regeneration_probability}
\end{equation}
where $Q_n(u)$ is the probability that the cycle has not regenerated by the end of time slot $n$.

Further, define the following quantities:
\begin{align}
\Gamma_{ik}(\bm m) &:= s\gamma^{m_i} e_i\bm P^{m_i}(I-\beta\bm P)^{-1}e_k^\top,\ i,k\in\mathcal X, \label{eq:persistent_regeneration_matrix}\\[1\jot]
\widetilde C_{i,K}(\bm m) &:= \sum_{n=1}^{m_i-1}\gamma^{n-1}g_i(n) \notag\\
&\qquad + \sum_{n=m_i}^{K}\gamma^{n-1}\eta^{n-m_i}\bigl(\lambda+\eta g_i(n)\bigr) \notag\\[1\jot]
&\qquad + \gamma^{m_i-1}\beta^{K+1-m_i}R_K, \qquad i\in\mathcal X, \label{eq:persistent_cycle_cost_cutoff}\\
C_{i,T}(u) &:= \sum_{n=1}^{T}\gamma^{n-1} \Big(\lambda u_nQ_{n-1}(u) + g_i(n)Q_n(u) \Big), \label{eq:finite_relaxed_block_cost}\\
\Gamma_T(u) &:= \sum_{n=1}^{T}\gamma^n s u_nQ_{n-1}(u) + \gamma^TQ_T(u), \label{eq:finite_regeneration_prob}\\
v_T &:= \min_{i\in\mathcal X,\,a\in\mathcal A_T} \frac{C_{i,T}(a)}{1-\Gamma_T(a)}.\label{eq:v_regenerative}
\end{align}
The following theorem gives a closed-form upper bound on the difference between the costs of the persistent and the optimal policy. The full derivation, which can be found in Appendix \ref{app:persistent_policy_gap_bound}, proceeds in three steps. First, we derive an upper bound on the normalized discounted cost of the persistent policy. Second, we derive a lower bound on the normalized discounted cost of the optimal policy. Third, we combine these two bounds to obtain an upper bound on the performance gap between the persistent policy and the optimal policy. All bounds are provided in closed form and can be computed from the parameters of the model.

\begin{theorem}
\label{thm:persistent_policy_gap_bound}
Fix a persistent first-attempt table $\bm m\in\mathbb N_{\geq1}^N$, a cutoff $K\geq\max_i m_i$, a finite regenerative horizon $T\geq1$, and a post-regeneration distribution $\bm\nu$\footnote{For an exact initial state $i_0$, one may take $\bm\nu=e_{i_0}$.}. Then
\begin{equation}
0 \leq \bar J_{\bm\nu}^{\rm pers}(\bm m) - \bar J_{\bm\nu}^* \leq B_{\bm\nu,K,T}(\bm m),
\label{eq:persistent_gap_bound}
\end{equation}
where
\begin{equation}
\label{eq:persistent_gap_bound_def}
\begin{aligned}
B_{\bm\nu,K,T}(\bm m) := \; &
(1-\gamma)\bm\nu^\top \bigl(I-\bm\Gamma(\bm m)\bigr)^{-1} \widetilde{\bm C}_K(\bm m) \\&- (1-\gamma)\sum_{i=1}^{N}\nu_i \min_{a\in\mathcal A_T} \left\{C_{i,T}(a)+\Gamma_T(a)v_T\right\}.
\end{aligned}
\end{equation}
\end{theorem}

Theorem \ref{thm:persistent_policy_gap_bound} provides an upper bound in closed form for the suboptimality of the persistent policy. Numerical experiments suggest that this bound is tight for most reasonable sets of parameters and the proposed heuristic performs well compared to the optimal solution while requiring lower computational efforts.

\section{Estimator with an Early Stationary Switch}\label{sec:hybrid_estimator}
We now focus on a complementary source of complexity: the implementation of the estimator itself. Recall that, after a successful transmission observing state $i\in\mathcal X$, and as long as no new successful update is received, the source marginal distribution at the $n$-th subsequent decision slot is
\[
p_i(n)=e_i\bm P^n.
\]
The corresponding MAP estimate is therefore
\[
\hat x_i(n) \in \arg\max_{j\in\mathcal X}(e_i\bm P^n)_j.
\]
An exact implementation of the MAP estimator requires, for every possible last observed state $i$, the storage or repeated computation of the entire sequence $\{\hat x_i(n)\}_{n\geq1}$. The main observation of this section is that, under standard ergodicity assumptions, this apparently infinite sequence has an eventually constant tail. Indeed, the source marginal distribution $e_i\bm P^n$ converges to the stationary distribution independently of the initial state $i$. If the stationary distribution has a unique most likely state, then this state eventually becomes the MAP estimate for every post-regeneration trajectory. We first formalize this property and show that the MAP estimator itself admits an exact representation through a table. We then consider an efficient implementation that switches to the stationary estimate before the exact MAP stabilization time. This early stationary switch further reduces the complexity of the estimation process, but may change the AoII cost. Finally, we derive an upper bound on the resulting difference in the optimal values. 

Assume that $\bm P$ is ergodic. Let $\bm\omega$ denote its unique stationary distribution, i.e.,
\[
\bm\omega\bm P=\bm\omega.
\]
We further assume that $\bm\omega$ has a unique mode, indicated by
\[
x_\omega = \arg\max_{j\in\mathcal X}\omega_j.
\]
The following lemma shows that $x_\omega$ is the exact MAP estimate after a sufficiently long no-success period.

\begin{lemma}
\label{lem:map_stabilization}
The MAP stabilization time
\[
\kappa_{\rm MAP} := \min\left\{ \kappa\in\mathbb Z_{\geq0}:
\hat x_i(n)=x_\omega,\; \forall i\in\mathcal X,\; \forall n>\kappa \right\}
\]
is finite.
\end{lemma}

\begin{proof}
See Appendix \ref{app:map_stabilization}.
\end{proof}

Lemma \ref{lem:map_stabilization} has an immediate implementation consequence. The monitor does not need to store the MAP estimate for arbitrarily large values of $n$. It is sufficient to store the finite set of estimates
\[
\left\{ \hat x_i(n): i\in\mathcal X,\; 1\leq n\leq\kappa_{\rm MAP}\right\},
\]
and to use $x_\omega$ for every time slot after $\kappa_{\rm MAP}$. This produces exactly the same estimate sequence as the original MAP rule. A further reduction can be obtained by switching to $x_\omega$ before the MAP sequence has necessarily stabilized. Fix an integer $\kappa\in\mathbb N_{\geq1}$ and define
\begin{equation}
\hat x_i^{\kappa}(n) =
\begin{cases}
\hat x_i(n), & 1\leq n\leq \kappa,\\
x_\omega, & n>\kappa.
\end{cases}
\label{eq:hybrid_estimator_definition}
\end{equation}
The estimator follows the MAP rule during the first $\kappa$ slots after a successful update and then uses the stationary mode. There are two different regimes. If $\kappa\geq\kappa_{\rm MAP}$, the switch occurs only after the MAP estimate has already stabilized, and \eqref{eq:hybrid_estimator_definition} is an exact implementation of the MAP estimator. If $\kappa<\kappa_{\rm MAP}$, the switch is performed earlier, and a smaller table is required.

For the estimator in \eqref{eq:hybrid_estimator_definition}, define
\[
\bm D_i^{\kappa}(n) =
\operatorname{diag}\!\big( \mathbbm{1}\{1\neq \hat x_i^{\kappa}(n)\}, \dots, \mathbbm{1}\{N\neq \hat x_i^{\kappa}(n)\} \big).
\]
The expected posterior AoII cost induced by this estimator in a no-success cycle starting from state $i$ is
\begin{equation}
g_i^{\kappa}(n) =
\sum_{r=1}^{n} e_i\bm P^r \bm D_i^{\kappa}(r) \bm P\bm D_i^{\kappa}(r+1) \cdots \bm P\bm D_i^{\kappa}(n)\mathbf 1,
\label{eq:hybrid_g_closed_form}
\end{equation}
where, for $r=n$, the product after $\bm D_i^{\kappa}(n)$ is omitted. We use the notation
\[
g_i^{\rm map}(n):=g_i(n), \qquad g_i^{\rm hyb}(n):=g_i^{\kappa}(n).
\]

\begin{remark}\label{rem:exact_and_early_stationary_switch}
If $\kappa\geq\kappa_{\rm MAP}$, then
\[
\hat x_i^\kappa(n)=\hat x_i(n), \qquad i\in\mathcal X,\quad n\geq1.
\]
Consequently,
\[
\bm D_i^\kappa(n)=\bm D_i(n), \qquad g_i^\kappa(n)=g_i(n), \qquad i\in\mathcal X,\quad n\geq1.
\]
In this regime, the stationary-mode representation is equivalent to the MAP estimation procedure.

If $\kappa<\kappa_{\rm MAP}$, the two estimate sequences can differ only during the finite interval
\[
\kappa<n\leq\kappa_{\rm MAP}.
\]
However, their AoII costs need not become identical immediately after $\kappa_{\rm MAP}$: the AoII at a given time slot depends not only on the current estimate, but also on the duration of the mismatch accumulated during the previous time slots. Therefore, a difference between the two estimators may affect the AoII cost even after their current estimates have become identical.
\end{remark}

We now quantify the effect of an early stationary switch on the optimal decision-making problem. For $E\in\{{\rm map},{\rm hyb}\}$, let $V_i^E(n)$ denote the optimal value of the reduced-state MDP when estimator $E$ is used, the latest successful observation was $i$, and the current time slot is the $n$-th time slot after that observation. The estimator does not affect the source marginal distribution, the channel outcomes, or the state observed after a successful transmission. Therefore, the MAP and hybrid models have the same transition probabilities, and they differ only in the expected AoII costs $g_i^E(n)$. Their Bellman equations can be written in the common form
\begin{equation}
\begin{aligned}
V_i^E(n) =
\min_{a\in\{0,1\}}\Big\{& (1-sa)g_i^E(n)+\lambda a\\[-1mm]&
+ \gamma(1-sa)V_i^E(n+1)\\[-1mm]&
+ \gamma sa\sum_{k=1}^{N} (\bm P^n)_{ik}V_k^E(1) \Big\}.
\label{eq:hybrid_reduced_bellman}
\end{aligned}
\end{equation}

If $\kappa\geq\kappa_{\rm MAP}$, the two reduced-state MDPs are identical and therefore
\[
V_{{\rm hyb},\kappa}^*(b_0) = V_{\rm map}^*(b_0).
\]
The remaining analysis is relevant when an arbitrary switching time is used, and in particular when $\kappa<\kappa_{\rm MAP}$.

The performance upper bound is constructed in three steps. First, we measure the difference between the two estimators within one time slot. Second, we identify states in which every optimal policy must transmit, thereby reducing the set of action sequences that must be considered. Third, we combine the discrepancies accumulated during successive regeneration cycles. To obtain a finite bound, it is useful to exclude action sequences that cannot be followed by an optimal policy. Intuitively, if the expected AoII cost is sufficiently large, remaining idle cannot be optimal: the possible saving of the transmission cost cannot compensate for the current and future AoII cost.

The following lemma uses this upper bound to give a sufficient condition under which a pull request is necessarily the optimal action.

\begin{lemma}
\label{lem:hybrid_active}
Let
\begin{equation}
U^{\rm AT} :=
\frac{\lambda+\eta(1-\beta)^{-1}}{1-\gamma}, \qquad \beta=\gamma\eta.
\label{eq:hybrid_always_transmit_upper_bound}
\end{equation}
For either estimator $E\in\{{\rm map},{\rm hyb}\}$, if
\begin{equation}
g_i^E(n) > \frac{\lambda}{s}+\gamma U^{\rm AT},
\label{eq:hybrid_forced_active_condition}
\end{equation}
then action $a=1$ is optimal at state $(i,n)$ in \eqref{eq:hybrid_reduced_bellman}.
\end{lemma}

\begin{proof}
See Appendix \ref{app:hybrid_active}.
\end{proof}

Lemma \ref{lem:hybrid_active} allows us to restrict the action sequences considered in the performance bound. In states satisfying \eqref{eq:hybrid_forced_active_condition}, an optimal policy must transmit. In the remaining states, either action may be selected.

For a switching time $\kappa\in\mathbb N_{\geq1}$ and a horizon $R\geq\kappa$, define the local AoII discrepancy
\begin{equation}
\varepsilon_i^{\kappa}(n) :=
\left|g_i(n)-g_i^{\kappa}(n)\right|,
\qquad i\in\mathcal X,\quad n\geq1.
\label{eq:hybrid_local_discrepancy}
\end{equation}
For $E\in\{{\rm map},{\rm hyb}\}$, let
\[
\begin{aligned}
\mathcal U_R^E:=\Bigl\{ \bm u\in\{0,1\}^{N\times R}:\;&
 u_i(n)=1\ \text{whenever}\\[-1mm]&
 g_i^E(n)>\frac{\lambda}{s}+\gamma U^{\rm AT},\\[-1mm]&
 i\in\mathcal X,\ n=1,\dots,R \Bigr\}.
\end{aligned}
\]
For $\bm u\in\mathcal U_R^E$, define
\[
Q_i^{\bm u}(0):=1, \qquad
Q_i^{\bm u}(n):=\prod_{r=1}^{n}\bigl(1-su_i(r)\bigr),
\quad n=1,\dots,R,
\]
and
\begin{equation}
\begin{aligned}
[\bm c_{\kappa,R}(\bm u)]_i := {}&
\sum_{n=1}^{R}\gamma^{n-1}Q_i^{\bm u}(n)\varepsilon_i^{\kappa}(n)\\[-1mm]
&+\gamma^RQ_i^{\bm u}(R)
\left(\frac{R+1}{1-\gamma}+\frac{\gamma}{(1-\gamma)^2}\right),
\end{aligned}
\label{eq:hybrid_cycle_discrepancy_vector}
\end{equation}
\begin{equation}
[\bm\Theta_R(\bm u)]_{ik}:=
\sum_{n=1}^{R}\gamma^n s u_i(n)Q_i^{\bm u}(n-1)(\bm P^n)_{ik}.
\label{eq:hybrid_regeneration_matrix}
\end{equation}

We can now state the resulting bound.

\begin{theorem}
\label{thm:hybrid_map_bound}
Fix $\kappa\in\mathbb N_{\geq1}$, and let $b_0$ be the synchronized initial belief in \eqref{eq:synchronized_initial_belief}, with initial state $i_0\in\mathcal X$. Let $V_{\rm map}^*(b_0)$ and $V_{{\rm hyb},\kappa}^*(b_0)$ be the optimal discounted values under the MAP and hybrid estimators, respectively. Then
\begin{equation}
\left|V_{\rm map}^*(b_0)-V_{{\rm hyb},\kappa}^*(b_0)\right| \leq
\gamma\inf_{R\geq\kappa} \max_E \; \bigl[B_{\kappa,R}^E(i_0)\bigr],
\label{eq:hybrid_initial_belief_value_bound}
\end{equation}
where
\begin{equation}
B_{\kappa,R}^E(i) :=
\max_{\bm u\in\mathcal U_R^E} \left[(I-\bm\Theta_R(\bm u))^{-1}\bm c_{\kappa,R}(\bm u)\right]_i.
\label{eq:hybrid_finite_horizon_certificate}
\end{equation}
\end{theorem}

\begin{proof}
See Appendix \ref{app:hybrid_map_bound}.
\end{proof}

Theorem \ref{thm:hybrid_map_bound} separates the two possible uses of the hybrid estimator. Choosing $\kappa\geq\kappa_{\rm MAP}$ gives an exact finite-memory implementation of the MAP rule and therefore zero value difference. Choosing $\kappa<\kappa_{\rm MAP}$ further reduces the length of the estimator table; in this case, \eqref{eq:hybrid_initial_belief_value_bound} quantifies the possible change in optimal discounted performance.

\section{Extension to Multiple Sources}
\label{sec:multi_source}
In this section, we extend the reduced single-source model to multiple independent source-monitor pairs sharing limited transmission resources. The exact state of each controlled arm, by Theorem \ref{thm:exact_reduction}, is the pair formed by the latest observed state and the elapsed no-success time. This keeps each arm countable and explicitly structured. The global state, however, is the Cartesian product of the individual states and therefore it grows exponentially with the number of sources.

To obtain a scalable scheduling rule, following standard works in the literature, we formulate the problem as a restless multi-armed bandit (RMAB), where each source is an arm. An arm is active when the monitor sends a pull request and passive otherwise. The arm is restless because its source distribution, MAP estimate, and expected AoII continue to evolve even when that arm is kept idle. The standard approach is based on Lagrangian relaxation. Instead of enforcing the activation constraint directly, we introduce a subsidy for passivity. The relaxation decouples the multi-source problem into independent single-arm problems. If the passive set of each relaxed arm grows monotonically with the subsidy, the arm is indexable and a Whittle index can be assigned to every reduced state. We first present this construction and a sufficient condition for indexability. We then describe the computation of the Whittle index and our proposed approximated policy.

\subsection{Restless Multi-Armed Bandit}
We consider $L$ independent sources. Source $j\in\{1,\dots,L\}$ has state space $\mathcal X^j=\{1,\dots,N_j\}$, transition matrix $\bm P^j$, and channel success probability $s^j\in[0,1]$. As in the single-source model, a synchronized initial state is handled separately. After the initial time slot, the source is represented by the reduced state variable
\begin{equation}
y_t^j=(i_t^j,n_t^j)\in\mathcal Y^j:=\mathcal X^j\times\mathbb N_{\geq1},
\label{eq:multisource_reduced_state}
\end{equation}
where $i_t^j$ is the latest observed source state and $n_t^j$ is the number of time slots elapsed since that observation. The action $a_t^j\in\{0,1\}$ indicates whether source $j$ is pulled. At most $M_{\rm act}<L$ sources can be active in one time slot:
\begin{equation}
\sum_{j=1}^{L}a_t^j\leq M_{\rm act},\qquad t\geq0.
\label{eq:multi_source_hard_constraint}
\end{equation}

For source $j$, let $g_i^j(n)$ be the expected AoII cost associated with its reduced state, computed from \eqref{eq:reliable_g_closed_form} using $\bm P^j$. The expected AoII cost is
\begin{equation}
d^j\bigl((i,n),a\bigr)=(1-s^j a)g_i^j(n).
\label{eq:multi_source_stage_cost}
\end{equation}
If no successful update occurs, the next state is $(i,n+1)$. If a request succeeds, state $k$ is revealed with probability $\bigl((\bm P^j)^n\bigr)_{ik}$ and the next reduced state is $(k,1)$. Hence, the multi-source scheduling problem is
\begin{mini!}
{\pi}
{\mathbb E^\pi\!\left[
\sum_{t=0}^{\infty}\gamma^t
\sum_{j=1}^{L}d^j(y_t^j,a_t^j)
\right]}
{\label{eq:multi_source_original_problem}}
{ }
\addConstraint{\sum_{j=1}^{L}a_t^j}{\leq M_{\rm act},\quad t\geq0.}
\end{mini!}
The policy observes the reduced state of every source and selects the active subset. Problem \eqref{eq:multi_source_original_problem} is an RMAB because every arm evolves under both actions.

\subsection{Lagrangian Relaxation and Whittle Index}

To decouple the arms, introduce an idle subsidy $W\in\mathbb R$. For a fixed source $j$, let $V_i^j(n,W)$ be the optimal relaxed value at state $(i,n)$. The passive and active action-value functions are
\begin{align}
Q_{0,i}^j(n,W) ={}& g_i^j(n)-W+\gamma V_i^j(n+1,W),
\label{eq:relaxed_single_arm_Q0}\\
Q_{1,i}^j(n,W) ={}& (1-s^j)g_i^j(n) + \gamma(1-s^j)V_i^j(n+1,W)\nonumber\\[-1mm]&
+ \gamma s^j\sum_{k=1}^{N_j}\bigl((\bm P^j)^n\bigr)_{ik}V_k^j(1,W).
\label{eq:relaxed_single_arm_Q1}
\end{align}
Thus,
\begin{equation}
V_i^j(n,W)=\min\bigl\{Q_{0,i}^j(n,W),Q_{1,i}^j(n,W)\bigr\}.
\label{eq:relaxed_single_arm_bellman}
\end{equation}

Define the action difference function
\begin{equation}
\Psi_i^j(n,W):=Q_{1,i}^j(n,W)-Q_{0,i}^j(n,W).
\label{eq:whittle_action_gap}
\end{equation}
Since costs are minimized, the passive action is optimal at $(i,n)$ if and only if $\Psi_i^j(n,W)\geq0$. Let
\begin{equation}
\mathcal Y_{\rm pass}^j(W) := \bigl\{(i,n)\in\mathcal Y^j:Q_{0,i}^j(n,W)\leq Q_{1,i}^j(n,W)\bigr\}
\label{eq:passive_set_definition}
\end{equation}
be the passive set under subsidy $W$. Source $j$ is indexable if, for every $W'>W$,
\begin{equation}
\mathcal Y_{\rm pass}^j(W)\subseteq\mathcal Y_{\rm pass}^j(W').
\label{eq:indexability_definition}
\end{equation}
When this property holds, the Whittle index of state $(i,n)$ is
\begin{equation}
\mathcal W^j(i,n) := \inf\bigl\{W:(i,n)\in\mathcal Y_{\rm pass}^j(W)\bigr\}.
\label{eq:whittle_index_definition}
\end{equation}
The Whittle index policy activates the $M_{\rm act}$ sources with the largest current indices.

We next provide a simple sufficient condition for indexability. It depends only on the discount factor and on the link reliability.

\begin{theorem}
\label{thm:indexability}
If
\begin{equation}
 \gamma\leq\frac{1}{1+s^j},
 \label{eq:multi_source_indexability_condition}
\end{equation}
then source $j$ is indexable.
\end{theorem}

\begin{proof}
See Appendix \ref{app:indexability}.
\end{proof}

\subsection{Online Whittle Index Computation}
The $j$-th relaxed arm has state space $\mathcal Y^j$. For offline computation, we use the finite reduced model of Section \ref{sec:model_reduction} with a selected truncation level $H$. Let
\begin{equation}
\mathcal Y_H^j=\mathcal X^j\times\{1,\dots,H\}
\label{eq:multi_source_computational_state_space}
\end{equation}
be the truncated state space for arm $j$. The same zero-continuation boundary is used for a null transition from $n=H$. For a fixed subsidy $W$, value iteration solves the finite relaxed Bellman equation. Under indexability, the passive set is monotone in $W$, so the index of each state can be found by bisection. The output of Algorithm \ref{alg:whittle_index_computation} is a table of indices for all states in $\mathcal Y_H^j$. The computation is repeated for each source, or only once for each class of identical sources.

\begin{algorithm}[t!]
\caption{Whittle Index (Source $j$)}
\label{alg:whittle_index_computation}
\begin{algorithmic}[1]
\State Choose $H$ and construct the finite reduced-state MDP on $\mathcal Y_H^j$
\State Choose tolerance $\varepsilon>0$ and maximum number of iterations $K_{\max}$
\For{all $y\in\mathcal Y_H^j$}
    \State Choose $W_{\min}$ and $W_{\max}$ such that $y\notin\mathcal Y^j_{\rm pass}(W_{\min})$ and $y\in\mathcal Y^j_{\rm pass}(W_{\max})$
    \State Set $\underline W(y)\gets W_{\min}$ and $\overline W(y)\gets W_{\max}$
    \For{$k=1,\dots,K_{\max}$}
        \State Set $W^{(k)}(y)\gets(\underline W(y)+\overline W(y))/2$
        \State Solve the finite relaxed Bellman equation using subsidy $W^{(k)}(y)$
        \If{$y\in\mathcal Y^j_{\rm pass}(W^{(k)}(y))$}
            \State $\overline W(y)\gets W^{(k)}(y)$
        \Else
            \State $\underline W(y)\gets W^{(k)}(y)$
        \EndIf
        \If{$\overline W(y)-\underline W(y)<\varepsilon$}
            \State \textbf{break}
        \EndIf
    \EndFor
    \State Set $\mathcal W^j(y)\gets(\underline W(y)+\overline W(y))/2$
\EndFor
\end{algorithmic}
\end{algorithm}

The online implementation requires only the index tables and the belief associated to each source. At the beginning of slot $t$, the monitor looks up one index per source and sends pull requests to the $M_{\rm act}$ sources with the highest index. After the transmission outcomes are observed, the beliefs are updated directly: a successful request revealing state $k$ starts a new no-success cycle at $(k,1)$, whereas an idle slot or a failed request increments the elapsed-duration coordinate. Algorithm \ref{alg:whittle_index_online} summarizes this procedure. Ties between equal indices are resolved according to a fixed deterministic ordering. Notice that the monitor never needs to observe the physical state or the instantaneous AoII of a source unless a requested transmission succeeds; a null observation covers both an idle source and an unsuccessful request and therefore produces the same update of the reduced state.

\begin{algorithm}[t!]
\caption{Online Whittle Index Policy}
\label{alg:whittle_index_online}
\begin{algorithmic}[1]
\State Initialize $y_0^j=(i_0^j,n_0^j)\in\mathcal Y^j$ from the latest successful observation, for every $j$
\For{$t=0,1,2,\dots$}
    \For{$j=1,\dots,L$}
        \State Retrieve $w_t^j\gets\mathcal W^j(y_t^j)$ from the index table
    \EndFor
    \State Sort $\{w_t^j\}_{j=1}^{L}$ in decreasing order
    \State Let $\mathcal J_t$ be the set of the $M_{\rm act}$ sources with the largest indices
    \For{$j=1,\dots,L$}
        \State Set $a_t^j\gets\mathbbm{1}\{j\in\mathcal J_t\}$
    \EndFor
    \State Apply $\bm a_t=(a_t^1,\dots,a_t^L)$ and collect the observations $\{o_t^j\}_{j=1}^{L}$
    \For{$j=1,\dots,L$}
        \If{$o_t^j=k$ for some $k\in\mathcal X^j$}
            \State Set $y_{t+1}^j\gets(k,1)$
        \Else
            \State Set $y_{t+1}^j\gets(i_t^j,n_t^j+1)$
        \EndIf
    \EndFor
\EndFor
\end{algorithmic}
\end{algorithm}

\subsection{Approximate Whittle Index Policy}
The offline computation of the Whittle index can be expensive because every state requires a subsidy bisection and every bisection step solves a Bellman equation. Furthermore, although simple, the sufficient condition provided in Theorem \ref{thm:indexability} can be restrictive, and indexability in general does not hold when that condition is not satisfied. To overcome these difficulties, we here propose an approximate Whittle index policy (AWIP) that computes exact indices only on representative reduced states and interpolates the remaining values. This heuristic can be used even under non-indexable regimes and proves to be well-performing in numerical simulations.

For source $j$, define the instantaneous update gain
\begin{equation}
\begin{aligned}
z^j(i,n) &
:= d^j\bigl((i,n),0\bigr) - d^j\bigl((i,n),1\bigr)\\&
= s^j g_i^j(n), \qquad (i,n)\in\mathcal Y_H^j.
\label{eq:awip_urgency_coordinate}
\end{aligned}
\end{equation}
This scalar is the expected immediate AoII reduction obtained by activating the source. Larger values correspond to states with larger expected AoII and/or a more reliable link.

Fix a number of anchors $K^j\leq|\mathcal Y_H^j|$. Select
\[
\mathcal G^j=\{y_1^j,\dots,y_{K^j}^j\}\subseteq\mathcal Y_H^j
\]
so that the coordinates $z^j(y_\ell^j)$ cover the range of $z^j$ over the finite state set. Write
\[
z_1^j\leq\cdots\leq z_{K^j}^j,
\qquad z_\ell^j:=z^j(y_\ell^j),
\]
with repeated coordinates removed. Exact indices are computed only at the anchor states. For a non-anchor state $y$ with $z^j(y)\in[z_\ell^j,z_{\ell+1}^j]$, define
\begin{align}
\widehat{\mathcal W}^j(y) ={}&
\mathcal W^j(y_\ell^j) + \frac{z^j(y)-z_\ell^j}{z_{\ell+1}^j-z_\ell^j}\nonumber\\[-1mm]&
\quad{} \left(\mathcal W^j(y_{\ell+1}^j) - \mathcal W^j(y_\ell^j) \right).
\label{eq:awip_linear_interpolation}
\end{align}
At an anchor state,
\begin{equation}
\widehat{\mathcal W}^j(y_\ell^j)=\mathcal W^j(y_\ell^j).
\label{eq:awip_anchor_exact}
\end{equation}
Thus, the number of exact bisection computations is reduced from $N^jH$ to $K^j$.

\begin{algorithm}[t!]
\caption{Approximate Whittle Index (Source $j$)}
\label{alg:aoii_grid_awip_computation}
\begin{algorithmic}[1]
\State Choose $H$ and construct $\mathcal Y_H^j$
\State Choose the number of anchors $K^j$
\For{all $y=(i,n)\in\mathcal Y_H^j$}
    \State Compute $z^j(y)=s^jg_i^j(n)$
\EndFor
\State Choose anchor states $\mathcal G^j=\{y_1^j,\dots,y_{K^j}^j\}$ as quantile states of $z^j(y)$
\State Sort the anchors so that $z_1^j\leq\cdots\leq z_{K^j}^j$
\For{all $y_\ell^j\in\mathcal G^j$}
    \State Compute $\mathcal W^j(y_\ell^j)$ through Algorithm \ref{alg:whittle_index_computation}
    \State Set $\widehat{\mathcal W}^j(y_\ell^j)\gets\mathcal W^j(y_\ell^j)$
\EndFor
\For{all $y\in\mathcal Y_H^j\setminus\mathcal G^j$}
    \If{$z^j(y)\leq z_1^j$}
        \State Set $\widehat{\mathcal W}^j(y)\gets\mathcal W^j(y_1^j)$
    \ElsIf{$z^j(y)\geq z_{K^j}^j$}
        \State Set $\widehat{\mathcal W}^j(y)\gets\mathcal W^j(y_{K^j}^j)$
    \Else
        \State Find $\ell$ such that $z_\ell^j\leq z^j(y)\leq z_{\ell+1}^j$
        \State Compute $\widehat{\mathcal W}^j(y)$ from \eqref{eq:awip_linear_interpolation}
    \EndIf
\EndFor
\end{algorithmic}
\end{algorithm}

The interpolation is intentionally one-dimensional. Distinct reduced states may have the same expected immediate update gain and are then assigned the same approximate index. This loss of information is the price paid for reducing the offline computation. Nevertheless, $z^j(i,n)=s^jg_i^j(n)$ is directly tied to the immediate benefit of a successful update and therefore gives an interpretable priority coordinate. In the indexable regime, the anchors carry exact Whittle indices and the AWIP is a low-complexity surrogate of the WIP. Outside the guaranteed indexable regime, the same construction can be used as a heuristic priority policy, provided the anchor values are interpreted accordingly.

No separate online algorithm is needed for the AWIP. Algorithm \ref{alg:whittle_index_online} is reused without changing the state tracking, scheduling, observation, or state-update steps; in the lookup step, $\mathcal W^j(y_t^j)$ is simply replaced by the approximate value $\widehat{\mathcal W}^j(y_t^j)$. Hence, the exact and approximate policies have the same online implementation and differ only in the index tables constructed offline.

\section{Numerical Results}\label{sec:numerical_results}
In this section, we present numerical results for both single-source and multi-source problems. First, we provide a visual representation of the optimal policy on the belief space. Second, we analyze the performance of our proposed persistent policy. Third, we compare the proposed multi-source index policies with a random scheduling baseline.

\subsection{Single-Source Experiments}
We first consider the single-source problem solved through the truncated MDP of Section \ref{sec:model_reduction}. The AoII truncation level is set to $H=25$, and the source has $N=5$ possible states. The discount factor, the success probability and the transmission cost are $\gamma=0.9$, $s=0.8$ and $\lambda=1.5$. We consider three different sources labeled \emph{volatile}, \emph{stable-a}, and \emph{stable-b}. The three classes represent different degrees of temporal persistence of the underlying Markov process. The transition matrices used for the three classes are
\[
\begin{aligned}
\bm P_{\mathrm{volatile}}
&=
\begin{bmatrix}
0.500 & 0.200 & 0.120 & 0.100 & 0.080\\
0.060 & 0.250 & 0.320 & 0.220 & 0.150\\
0.180 & 0.080 & 0.400 & 0.220 & 0.120\\
0.280 & 0.120 & 0.120 & 0.300 & 0.180\\
0.220 & 0.240 & 0.080 & 0.180 & 0.280
\end{bmatrix},
\\[0.5em]
\bm P_{\mathrm{stable\text{-}a}}
&=
\begin{bmatrix}
0.920 & 0.035 & 0.020 & 0.015 & 0.010\\
0.010 & 0.920 & 0.035 & 0.020 & 0.015\\
0.015 & 0.040 & 0.920 & 0.010 & 0.015\\
0.030 & 0.010 & 0.035 & 0.920 & 0.005\\
0.025 & 0.010 & 0.015 & 0.030 & 0.920
\end{bmatrix},
\\[0.5em]
\bm P_{\mathrm{stable\text{-}b}}
&=
\begin{bmatrix}
0.870 & 0.055 & 0.035 & 0.025 & 0.015\\
0.015 & 0.870 & 0.055 & 0.035 & 0.025\\
0.030 & 0.050 & 0.870 & 0.030 & 0.020\\
0.045 & 0.020 & 0.050 & 0.870 & 0.015\\
0.060 & 0.015 & 0.030 & 0.025 & 0.870
\end{bmatrix}.
\end{aligned}
\]

Each point in Figures \ref{fig:single_source_volatile}--\ref{fig:single_source_stable-b} corresponds to a reachable belief $b\in\mathcal B_{D,H}$ generated by the truncated recursion \eqref{eq:belief_recursion}. Since a belief is a distribution over the source state and the AoII, it cannot be represented directly in two dimensions. We therefore project each belief onto the pair $\left(\bar\Delta(b),p_{\max}(b)\right)$, with
\begin{equation}
\bar\Delta(b):=\sum_{i=1}^{N}\sum_{\Delta=0}^{D}\Delta b(i,\Delta),
\qquad
p_{\max}(b):=\max_{i\in\{1,\dots,N\}}\pi_b(i).
\label{eq:numerical_projection_coordinates}
\end{equation}
The first coordinate is the expected prior AoII. The second coordinate measures the confidence of the MAP estimate. The marker associated with each belief represents the optimal action
\[
a^\star(b)\in\arg\min_{a\in\{0,1\}}Q^\lambda(b,a).
\]
The range of $Q^\lambda(b,1)-Q^\lambda(b,0)$ is also reported in each figure: negative values favor pull requests and positive values favor idling.

\begin{figure}[t!]
    \centering
    \includegraphics[width=0.5\columnwidth]{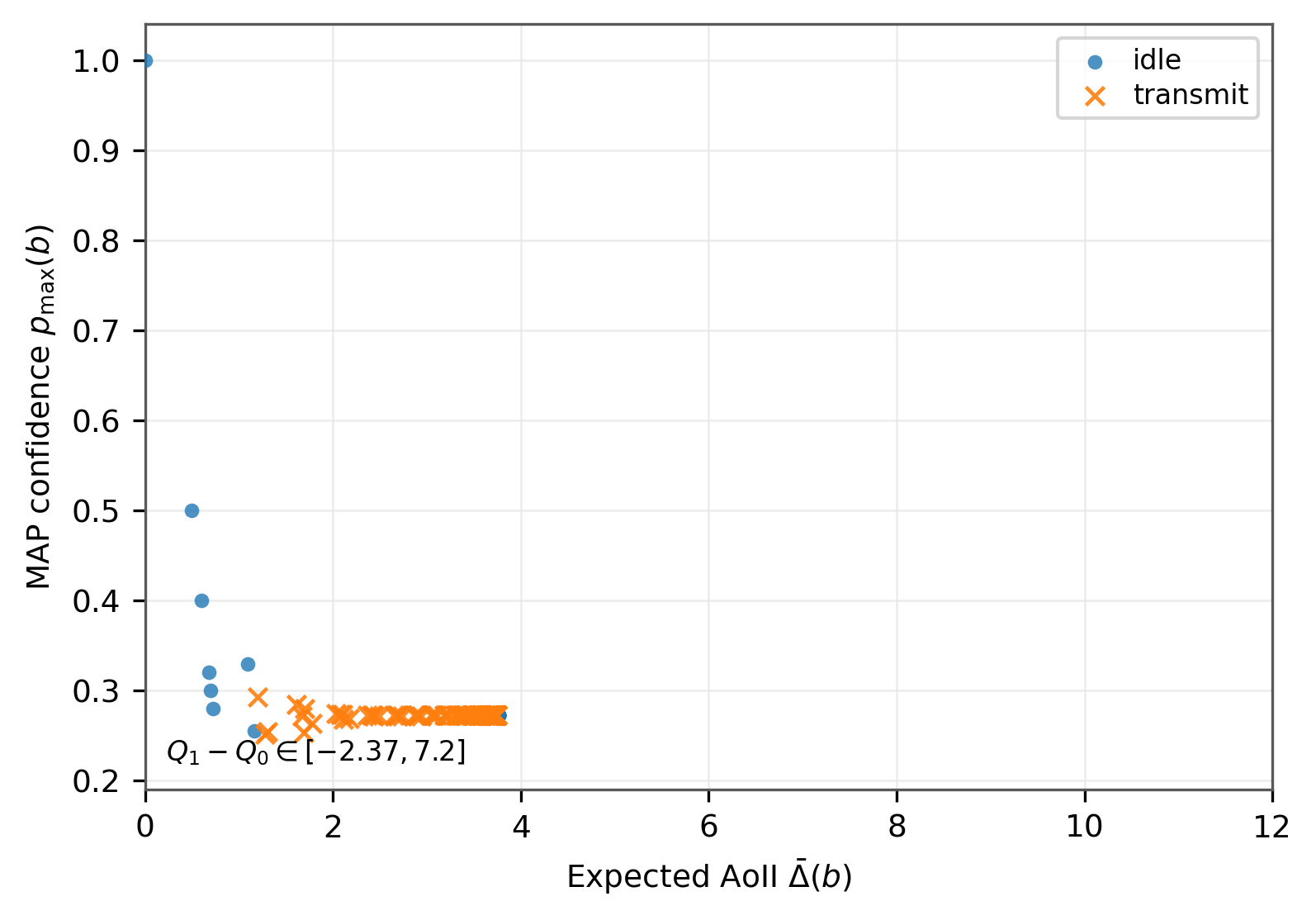}
    \caption{Projection of the optimal single-source action over reachable beliefs for the \emph{volatile} source.}
    \label{fig:single_source_volatile}
\end{figure}

\begin{figure}[!t]
    \centering
    \includegraphics[width=0.5\columnwidth]{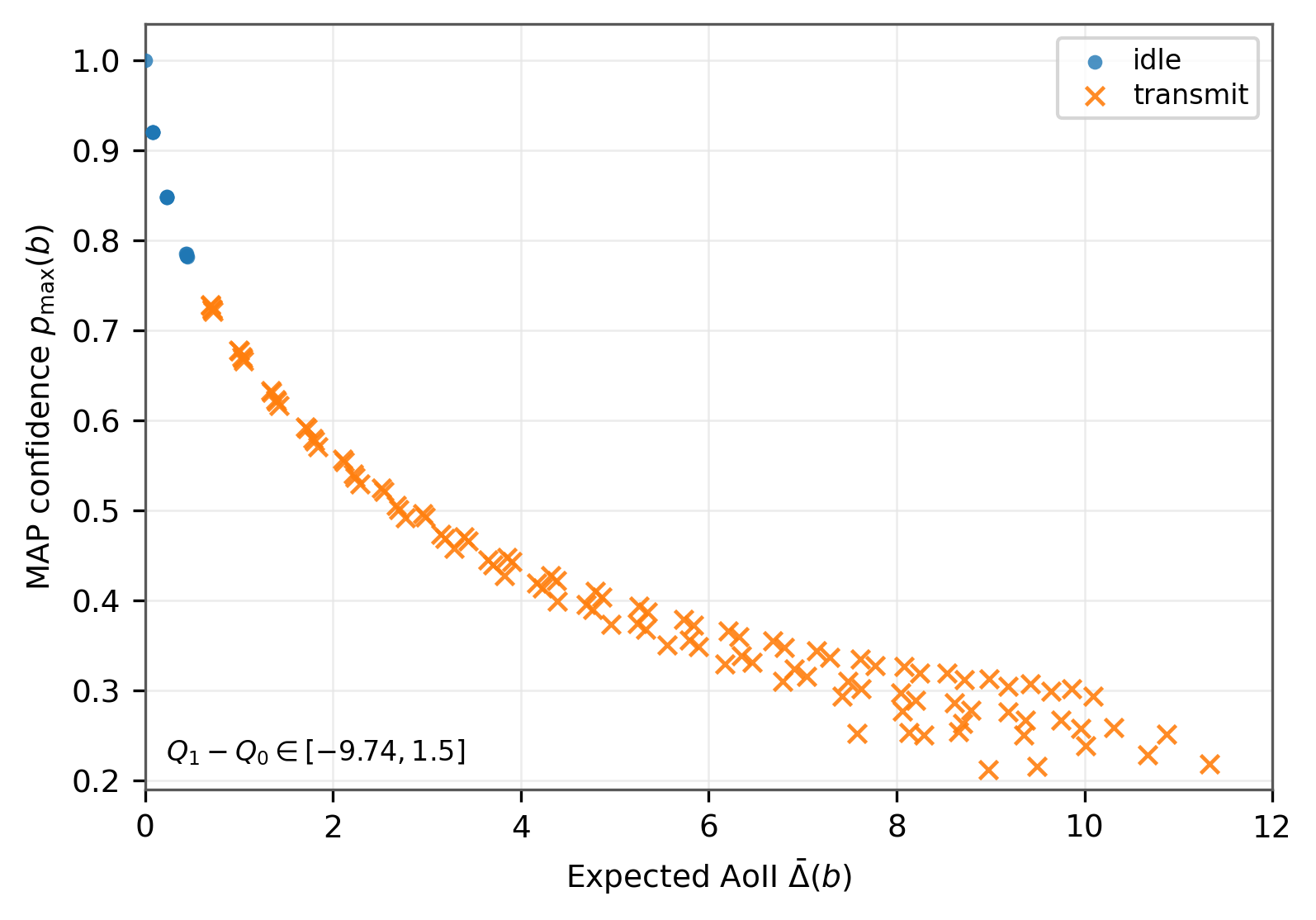}
    \caption{Projection of the optimal single-source action over reachable beliefs for the \emph{stable-a} source.}
    \label{fig:single_source_stable-a}
\end{figure}

\begin{figure}[t!]
    \centering
    \includegraphics[width=0.5\columnwidth]{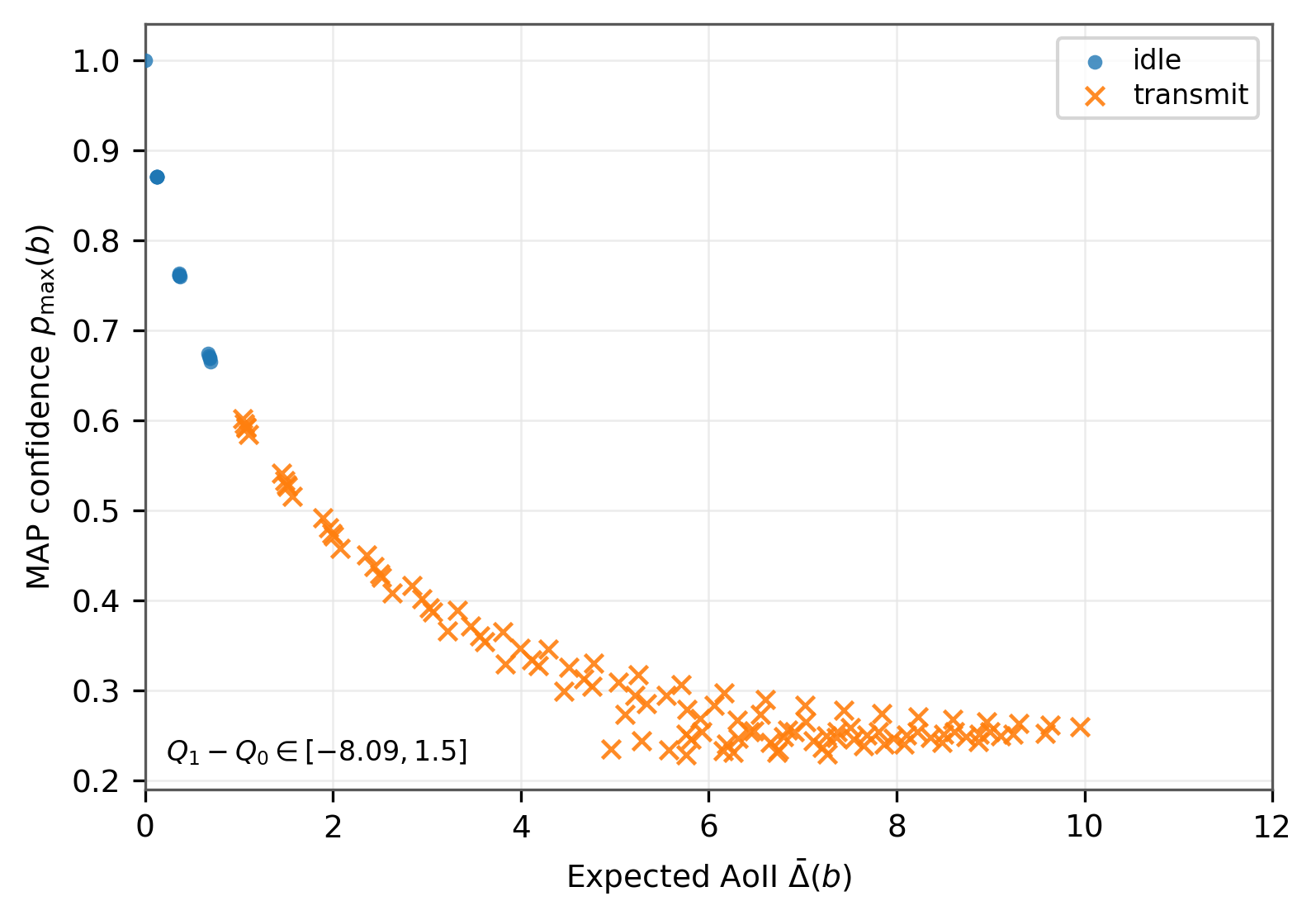}
    \caption{Projection of the optimal single-source action over reachable beliefs for the \emph{stable-b} source.}
    \label{fig:single_source_stable-b}
\end{figure}

Figures \ref{fig:single_source_volatile}--\ref{fig:single_source_stable-b} show that the optimal policy has an interpretable switching behavior in the projected belief space. Beliefs with low expected AoII and high MAP confidence are typically kept idle, because the monitor is already likely to be synchronized with the source and the immediate value of an update is small. As $\bar\Delta(b)$ increases, requesting a transmission becomes optimal over most of the reachable set, since a successful update can reset the AoII and improve the future estimate. The two stable sources exhibit a clean separation between idle and transmit regions. This happens because, under more persistent transition dynamics, the reachable beliefs tend to evolve along a more regular low-dimensional curve: as the expected AoII increases, the MAP confidence decreases smoothly, and the optimal action changes in an almost threshold-like way. In the volatile case, the reachable beliefs are more compressed in the low-confidence region, and the two-dimensional projection hides more information about the full belief. It is important to stress that Figures \ref{fig:single_source_volatile}--\ref{fig:single_source_stable-b} are projections of the belief-MDP policy, not a proof that the optimal policy is characterized by a scalar threshold. Two beliefs with similar values of $\bar\Delta(b)$ and $p_{\max}(b)$ may have different distributions over the non-MAP states and over the AoII values, and this residual information can affect the continuation value. Nevertheless, the numerical results indicate that the expected AoII is a strong explanatory coordinate for the optimal decision, while $p_{\max}(b)$ helps interpret how the uncertainty of the MAP estimator changes across the reachable beliefs.

In a second analysis, we also validate the persistent policy described in Section \ref{sec:unreliable_persistent_policy} in the same single-source setting. For each source, the waiting-time vector $m=(m_1,\ldots,m_N)$ is computed offline and then used online as a simple table: after a successful observation of state $i$, the monitor remains idle until time reaches $m_i$, and then transmits persistently until the next successful update. Table \ref{tab:single_source_persistent_validation} compares the resulting persistent policy with the optimal policy and with an always-transmit policy, which transmits at every time step. The persistent policy remains close to the optimal truncated policy for all sources, while requiring only a waiting time table rather than the full belief-MDP action table. The relative loss with respect to the simulated optimal policy is negligible. These results confirm the high performance of our proposed heuristic as a suboptimal alternative to the optimal policy.

\begin{table}[ht]
\centering
\caption{Single-source validation of the persistent policy.}
\label{tab:single_source_persistent_validation}
\resizebox{0.5\columnwidth}{!}{%
\begin{tabular}{llrrr}
\toprule
Source class & Policy & $J$ & $\bar J$ & 95\% CI half-width for $J$ \\
\midrule
\emph{stable-a} & Persistent & 5.2938 & 0.5294 & 0.0091 \\
\emph{stable-a} & Optimal & 5.2929 & 0.5293 & 0.0092 \\
\emph{stable-a} & Always transmit & 13.7090 & 1.3709 & 0.0008 \\
\midrule
\emph{stable-b} & Persistent & 6.3801 & 0.6380 & 0.0091 \\
\emph{stable-b} & Optimal & 6.3793 & 0.6379 & 0.0091 \\
\emph{stable-b} & Always transmit & 13.8329 & 1.3833 & 0.0013 \\
\midrule
\emph{volatile} & Persistent & 10.9175 & 1.0918 & 0.0081 \\
\emph{volatile} & Optimal & 10.9122 & 1.0912 & 0.0083 \\
\emph{volatile} & Always transmit & 14.8665 & 1.4867 & 0.0026 \\
\bottomrule
\end{tabular}
}
\end{table}

\subsection{Multi-Source Experiments}
We now evaluate the multi-source scheduling rules. We compare the Whittle index policy (WIP), the approximate Whittle index policy (AWIP), and a random policy that selects the active sources uniformly at random. We use the same parameter set as in the single-source analysis, with $M_{\mathrm{act}}=2$ sources to be selected at each time step, out of $L=9$ total sources assigned equally to the $3$ stability classes described in the single-source case. The performance metric is the average normalized discounted cost per source:
\begin{equation}
\bar J := \frac{1-\gamma}{L}\,\mathbb E\Big[\sum_{t=0}^{\infty}\gamma^t\sum_{j=1}^{L}d^j(b_t^j,a_t^j)\Big],
\label{eq:numerical_normalized_multisource_cost}
\end{equation}
where the factor $1-\gamma$ puts the discounted cost on the same scale as the cost function. For the AWIP, the approximate index table is built using $10$ anchor beliefs, while the exact index table contains $122$ beliefs. The infinite-horizon sum is approximated by truncating the horizon at $K_{\max}=200$, and every experiment is repeated $10000$ times. We first fix $\gamma=0.55$ to satisfy the sufficient condition for indexability described in Theorem \ref{thm:indexability}. We then fix $\gamma=0.9$ to analyze the case beyond the indexability condition.

\begin{table}[ht]
\centering
\caption{Multi-source policy comparison in the indexable regime.}
\label{tab:multisource_indexable}
\resizebox{0.5\columnwidth}{!}{%
\begin{tabular}{lcccc}
\toprule
Policy & $J$ & $\bar J$ & $95\%$ CI half-width for $J$ & Reduction vs. random \\
\midrule
WIP    & $3.7505$ & $0.1875$ & $0.0032$ & $25.34\%$ \\
AWIP   & $3.7529$ & $0.1876$ & $0.0032$ & $25.29\%$ \\
Random & $5.0232$ & $0.2512$ & $0.0078$ & $--$ \\
\bottomrule
\end{tabular}
}
\end{table}

\begin{table}[ht]
\centering
\caption{Multi-source policy comparison beyond the indexable regime.}
\label{tab:multisource_non_indexable}
\resizebox{0.5\columnwidth}{!}{%
\begin{tabular}{lcccc}
\toprule
Policy & $J$ & $\bar J$ & $95\%$ CI half-width for $J$ & Reduction vs. random \\
\midrule
AWIP   & $51.4835$ & $0.5720$ & $0.0662$ & $38.23\%$ \\
Random & $83.3463$ & $0.9261$ & $0.1452$ & $--$ \\
\bottomrule
\end{tabular}
}
\end{table}

Table \ref{tab:multisource_indexable} shows that the two index-based policies achieve almost identical performance in this experiment. The small numerical difference between WIP and AWIP is negligible. Both policies reduce the normalized discounted cost by about $25\%$ with respect to the random policy. Table \ref{tab:multisource_non_indexable} shows that the proposed AWIP performs significantly better than the baseline even beyond the indexable regime. These results support the use of AWIP as a low-complexity alternative for WIP. In this experiment, AWIP uses only $30$ anchor beliefs instead of computing an exact index for all $366$ reachable belief states. Therefore, the online implementation remains identical to WIP, since both policies simply rank the current sources by their indices, but the offline construction of the approximate table is substantially reduced in terms of the number of states where the exact index computation is required.

\section{Conclusions}\label{sec:conclusion}
In this paper, we studied pull-based remote state estimation of a Markovian source under an AoII performance criterion. Since the monitor cannot observe the source state nor the instantaneous AoII, we first formulated the problem through the joint source-AoII belief and derived the corresponding discounted belief-MDP. We then exploited the particular information structure of the model. Because a successful transmission reveals the state exactly and a null observation is uninformative about the source, every reachable belief is determined by the latest observed state and the elapsed no-success duration. This removes the need to operate on a generic belief simplex. We introduced a finite approximation by truncating the elapsed duration, derived an explicit bound on the resulting error, and provided a principled rule for selecting the truncation level. For reliable links, we showed that an optimal policy is completely described by a table of waiting times. For unreliable links, we proposed a persistent policy and derived computable performance certificates based on regenerative upper and lower bounds. We also established finite-time stabilization of the MAP estimate and introduced a hybrid estimator that replaces the transient tail by a stationary estimate, together with a bound on the resulting change in optimal value. Finally, we extended the exact reduced representation to multiple sources and formulated the scheduling problem as a restless multi-armed bandit. We derived a sufficient condition for indexability and proposed an approximate Whittle index policy based on interpolation over a small set of anchor states. The resulting framework separates the hidden physical dynamics from the observable control state and provides exact structural reductions together with finite, controllable numerical approximations.

\bibliographystyle{IEEEtran}
\bibliography{references}

\appendices

\section{Proof of Theorem \ref{thm:exact_reduction}}
\label{app:exact_reduction}
We first prove that the set $\{b_{i,n}\}$ is closed under every possible observation. Suppose that the current belief is $b_{i,n}$. If a transmission succeeds and reveals state $k$, the posterior belief is the synchronized belief $b^k$. Therefore,
\[
F(b_{i,n},k)=F(b^k,k)=b_{k,1}.
\]
If no successful transmission occurs, the observation is $\varnothing$. Under the idle action this observation occurs with probability one, whereas under the active action its probability is $1-s$. The null observation conveys no additional information about the source. Furthermore, for every reachable belief, the mass associated with the current MAP estimate already has zero AoII, so that $U(b_{i,n},\varnothing)=b_{i,n}$. Consequently, we have $F(b_{i,n},\varnothing)=b_{i,n+1}$. At the synchronized initial belief $b_0$, the only possible observed state is $i_0$, and both a null observation and a successful observation lead, after the source transition, to $b_{i_0,1}$. The two identities above imply by induction that every reachable prior belief other than $b_0$ is of the form $b_{i,n}$. The pair $(i,n)$ is observable because it records the latest observed state and the subsequent time without a success.

We next derive the source marginal distribution. Let
\[
\pi_{i,n}(x):=\sum_{\Delta=0}^{\infty}b_{i,n}(x,\Delta).
\]
From the definition of $b_{i,1}$,
\[
\pi_{i,1}(x)=P_{ix}=(e_i\bm P)_x.
\]
Since a null observation does not change the a posteriori source distribution, the prediction step in \eqref{eq:belief_recursion} gives
\[
\begin{aligned}
\pi_{i,n+1}(x)
&=\sum_{m=1}^{N}\sum_{\delta=0}^{\infty}
 b_{i,n}(m,\delta)P_{mx}\\
&=\sum_{m=1}^{N}\pi_{i,n}(m)P_{mx}.
\end{aligned}
\]
Thus, by induction,
\[
\pi_{i,n}=e_i\bm P^n=p_i(n),
\]
and the corresponding MAP estimate is $\hat x_i(n)$ in \eqref{eq:reduced_map_estimate}.

It remains to verify the expression for the expected AoII. Consider a no-success cycle that starts immediately after state $i$ is observed, and let $X_h$ denote the source state in its $h$-th subsequent time slot. In the absence of another successful update, the estimate in time slot $h$ is $\hat x_i(h)$. Since the AoII is reset whenever $X_h=\hat x_i(h)$, its posterior value in time slot $n$ satisfies the identity
\begin{equation}
\label{eq:path_identity}
\Delta_n^+ =
\sum_{r=1}^{n} \prod_{h=r}^{n} \mathbbm{1}\{X_h\neq \hat x_i(h)\},
\end{equation}
which counts the number of consecutive mismatched slots ending at time slot $n$. Taking expectations conditional on the latest observed state being $i$ yields
\[
\begin{aligned}
g_i(n)
&=\sum_{r=1}^{n}
\Pr_i\!\left(
X_h\neq \hat x_i(h),\ h=r,\ldots,n
\right)\\
&=\sum_{r=1}^{n}
 e_i\bm P^r\bm D_i(r)\bm P\bm D_i(r+1)
 \cdots \bm P\bm D_i(n)\mathbf 1,
\end{aligned}
\]
where the second equality follows from the Markov property and the definition of the mismatch matrices. This proves \eqref{eq:reliable_g_closed_form}. Finally, the identity \eqref{eq:path_identity} gives $0\leq \Delta_n^+\leq n$, and hence $0\leq g_i(n)\leq n$.

Therefore, the full reachable belief, the current expected cost, and the conditional law of the next belief under either action are all determined by the observable pair $(i,n)$. This proves that $(i,n)$ is a sufficient state descriptor.\qed

\section{Proof of Theorem \ref{thm:horizon_truncation}}
\label{app:horizon_truncation}
We first specify the finite model. Inside $\mathcal Y_H$, the costs and transitions coincide with those of the exact reduced-state MDP. For $n<H$,
\begin{equation}
\begin{aligned}
T_H\big((i,n+1)|(i,n),a\big)&=1-sa,\\
T_H\big((k,1)|(i,n),a\big)&=sa(\bm P^n)_{ik},
\label{eq:finite_reduced_transitions_interior}
\end{aligned}
\end{equation}
while at the boundary
\begin{equation}
\begin{aligned}
T_H\big(\partial|(i,H),a\big)&=1-sa,\\
T_H\big((k,1)|(i,H),a\big)&=sa(\bm P^H)_{ik}.
\label{eq:finite_reduced_transitions_boundary}
\end{aligned}
\end{equation}
The terminal state $\partial$ has zero stage cost and zero continuation value. Therefore, if $V_{i,H}(n)$ is the finite-model value, its Bellman equation is
\begin{align}
V_{i,H}(n)=\min_{a\in\{0,1\}}\Big\{&
(1-sa)g_i(n)+\lambda a\nonumber\\
&+\gamma(1-sa)\mathbbm{1}\{n<H\}V_{i,H}(n+1)\nonumber\\
&+\gamma sa\sum_{k=1}^{N}(\bm P^n)_{ik}V_{k,H}(1)
\Big \},
\label{eq:finite_reduced_bellman}
\end{align}
and the value from the synchronized initial condition is
\begin{equation}
V_H^*(b_0):=\gamma V_{i_0,H}(1).
\label{eq:finite_reduced_initial_value}
\end{equation}

We next construct a uniform bound on the continuation omitted at the boundary. A positive AoII can persist for $m$ consecutive time slots only if the source avoids the estimate during that whole time interval. The conditional probability of avoiding any prescribed estimate in one transition is at most $q$. Therefore,
\begin{equation}
g_i(n)\leq A_q(n):=\sum_{m=1}^{n}q^m.
\label{eq:g_markov_bound}
\end{equation}

To bound the continuation beyond the truncation boundary, we consider a policy that transmits in every time slot until an update succeeds and continues to use the same rule after every subsequent success. Since the problem is a cost minimization problem, the optimal continuation value cannot exceed the cost incurred by this always-transmit policy.

Under the always-transmit policy, the current cycle ends at the first successful update. If the first $r$ transmission attempts fail, the probability of reaching the next attempt is $\eta^r$, while its cost is discounted by $\gamma^r$. The resulting discounted factor is therefore $\beta^r=(\gamma\eta)^r$. After that, the transmission cost $\lambda$ is always paid, while the AoII cost is paid only if the attempt also fails and is bounded by $\eta A_q(n+r)$. Hence, starting from a no-success duration $n$, the discounted cost accumulated before the next successful update is bounded by
\begin{equation}
\sum_{r=0}^{\infty}\beta^r\bigl(\lambda+\eta A_q(n+r)\bigr) = \overline C_n.
\label{eq:Cbar_AT_series_representation}
\end{equation}

The discounted weight associated with the new cycle that starts after the next successful update is
\begin{equation}
\Gamma^{\rm AT} := \sum_{r=0}^{\infty}\gamma^{r+1}\eta^r s = \frac{\gamma s}{1-\beta}.
\label{eq:finite_AT_regeneration_weight}
\end{equation}
Indeed, the event that the next success occurs after exactly $r$ failed attempts has probability $\eta^r s$, and the continuation value of the new cycle is discounted by $\gamma^{r+1}$.

To obtain the closed form in \eqref{eq:Cbar_AT_closed_form}, first consider $0\leq q<1$, for which
\[
A_q(n+r)=\sum_{m=1}^{n+r}q^m=\frac{q\bigl(1-q^{n+r}\bigr)}{1-q}.
\]
Therefore,
\begin{align*}
\overline C_n &=
\lambda\sum_{r=0}^{\infty}\beta^r+\frac{\eta q}{1-q} \left(\sum_{r=0}^{\infty}\beta^r-q^n\sum_{r=0}^{\infty}(\beta q)^r\right)\\&
= \frac{\lambda}{1-\beta}+\frac{\eta q}{1-q} \left (\frac{1}{1-\beta}-\frac{q^n}{1-\beta q}\right).
\end{align*}
For $q=1$, $A_1(n+r)=n+r$, and hence
\begin{align*}
\overline C_n &=
(\lambda+\eta n)\sum_{r=0}^{\infty}\beta^r+\eta\sum_{r=0}^{\infty}r\beta^r\\&
= \frac{\lambda+\eta n}{1-\beta}+\frac{\eta\beta}{(1-\beta)^2}.
\end{align*}
These are exactly the two cases in \eqref{eq:Cbar_AT_closed_form}.

It remains to account for the continuation after the first successful update. After a success, the process regenerates at some state $(k,1)$. Let $W_k^{\rm AT}$ be the always-transmit value starting from $(k,1)$ and define
\[
W^{\rm AT}:=\max_{k\in\mathcal X}W_k^{\rm AT}.
\]
The cost incurred before the following success is bounded by $\overline C_1$, while the continuation after that success has discounted weight at most $\Gamma^{\rm AT}$. Therefore,
\[
W^{\rm AT}\leq \overline C_1+\Gamma^{\rm AT}W^{\rm AT},
\]
which gives the uniform bound
\begin{equation}
W^{\rm AT}\leq\overline W:=\frac{\overline C_1}{1-\Gamma^{\rm AT}}.
\label{eq:Wbar_AT}
\end{equation}

The first state omitted by the truncation is $(i,H+1)$. Starting from this state, $\overline C_{H+1}$ bounds the remainder of the current no-success cycle. If this cycle ends with a successful update, all subsequent regeneration cycles contribute at most $\overline W$, with discounted weight $\Gamma^{\rm AT}$. Hence the entire continuation from the first omitted state is bounded by
\begin{align*}
\overline C_{H+1}+\Gamma^{\rm AT}\overline W &=
\overline C_{H+1}+\frac{\gamma s}{1-\gamma}\overline C_1\\&
= \overline M_H,
\end{align*}
where the first equality follows from $1-\Gamma^{\rm AT}=(1-\gamma)/(1-\beta)$. The first term covers the remainder of the current cycle, while the second covers all subsequent regeneration cycles.

All stage costs are nonnegative, and the only difference between the two models is that the finite one discards this continuation after a null boundary transition. Thus,
\[
V_H^*(b_0)\leq V^*(b_0).
\]

Let $\phi_H^*$ be optimal for the finite model and apply it in the exact MDP until the first time $\tau_H$ at which a null transition from $(i,H)$ reaches $(i,H+1)$. From that time onwards, transmit in every time slot according to the always-transmit policy described above. The two systems have identical states, actions, and costs before $\tau_H$. Moreover, the first omitted state cannot be reached from $b_0$ before time $H+1$, so $\tau_H\geq H+1$ almost surely.

At the exit time, the exact state is $(i,H+1)$ for some $i$, and its continuation under the always-transmit rule is at most $\overline M_H$. The extended policy is feasible for the exact problem, and therefore
\begin{align*}
V^*(b_0)
&\leq V_H^*(b_0)+\mathbb E\big[\gamma^{\tau_H}\big]\overline M_H\\
&\leq V_H^*(b_0)+\gamma^{H+1}\overline M_H.
\end{align*}
Combining this inequality with $V_H^*(b_0)\leq V^*(b_0)$ proves the claim.\qed

\section{Proof of Theorem \ref{thm:reliable_waiting_table}}
\label{app:reliable_waiting_table}
The general Bellman equation \eqref{eq:general_reduced_bellman} becomes
\begin{equation}
\begin{aligned}
V_i(n)=\min\Big \{& g_i(n)+\gamma V_i(n+1),\\[-1mm]&
\lambda + \gamma\sum_{k=1}^{N}(\bm P^n)_{ik}V_k(1) \Big\},
\end{aligned}\label{eq:reliable_bellman}
\end{equation}
for $i\in\mathcal X$ and  $n\geq1$. The first branch corresponds to waiting one more time slot. The second branch transmits immediately, pays only the sampling cost in the current time slot, and starts a new cycle after observing state $k$. Let
\[
W_i=V_i(1), \qquad i\in\mathcal X,
\]
be the optimal value at the first decision slot after a successful observation of state $i$. Starting from $(i,1)$, any deterministic policy either transmits immediately, waits for some number of idle slots and then transmits, or never transmits. Hence, for a proposed vector $W\in\mathbb R^N$, define, for $m\in\mathbb N_{\geq1}$,
\begin{equation}
F_i(m,W) = \sum_{n=1}^{m-1}\gamma^{n-1}g_i(n) + \gamma^{m-1}\lambda + \gamma^m\sum_{k=1}^{N}(\bm P^m)_{ik}W_k,
\label{eq:reliable_F_finite}
\end{equation}
and
\[
F_i(\infty,W) = \sum_{n=1}^{\infty}\gamma^{n-1}g_i(n).
\]
The infinite-waiting cost is finite because $g_i(n)\leq n$ and $\gamma\in(0,1)$. Define the Bellman operator $\mathcal T:\mathbb R^N\to\mathbb R^N$ by
\begin{equation}
(\mathcal T W)_i =
\min_{m\in\mathbb N_{\geq1}\cup\{\infty\}} F_i(m,W).
\label{eq:reliable_operator}
\end{equation}
For any $W,W'\in\mathbb R^N$,
\begin{align}
|(\mathcal T W)_i-(\mathcal T W')_i|& \leq
\sup_{m\geq1} \gamma^m \left| \sum_{k=1}^{N}(\bm P^m)_{ik}(W_k-W'_k)\right| \\& \leq
\gamma\|W-W'\|_\infty .
\label{eq:reliable_contraction}
\end{align}
Thus, $\mathcal T$ is a contraction with modulus $\gamma$. By the Banach fixed point theorem, there exists a unique fixed point $W^*$ satisfying
\begin{equation}
W^* = \mathcal T W^*.
\label{eq:reliable_fixed_point}
\end{equation}
For each $i\in\mathcal X$, choose
\begin{equation}
m_i^* \in \arg\min_{m\in\mathbb N_{\geq1}\cup\{\infty\}} F_i(m,W^*).
\label{eq:reliable_m_star}
\end{equation}
The minimum is well-defined, because $F_i(m,W^*)\to F_i(\infty,W^*)$ as $m\to\infty$, and the point $m=\infty$ is included as a feasible waiting time.

The policy induced by $\bm m^*$ remains idle at time steps $n<m_i^*$ after observing $i$ and transmits at time step $n=m_i^*$. By construction, its value satisfies \eqref{eq:reliable_fixed_point}, and the explicit form \eqref{eq:reliable_F_finite} is exactly the Bellman recursion \eqref{eq:reliable_bellman} along one regeneration cycle. Therefore, the induced policy reaches the optimal value from every reduced state.\qed

\section{Proof of Theorem \ref{thm:persistent_policy_gap_bound}}\label{app:persistent_policy_gap_bound}
The proof is structured in two steps. Decompose the upper bound as $B_{\bm\nu,K,T}(\bm m) = U_{\bm\nu, K}(\bm m) - L_{\bm\nu,T}$. First, we prove that $U_{\bm\nu, K}(\bm m)$ is an upper bound on the normalized cost of the persistent policy. The argument uses the regenerative structure induced by successful transmissions. Starting from a state observed after a successful update, we separate the cost accumulated before the next success from the discounted value of the system after the next success; this gives a linear regenerative equation. The only non-computable term is the infinite tail of the cycle cost: we handle this term by fixing a cutoff and by bounding the AoII cost after the cutoff with the inequality $g_i(n)\leq n$ from Theorem \ref{thm:exact_reduction}. Second, we prove that $L_{\bm\nu,T}$ is a lower bound on the normalized cost of the optimal policy.  We derive it through a regenerative relaxation, which is independent of the persistent table $\bm m$ and can be used as a benchmark for other policies as well. To prove this result, we relate the optimal post-regeneration value to a relaxed problem with value $V_{\min}^*$; then, we show that the resulting relaxation reaches its minimum at a deterministic action sequence. The difference between the two quantities thus obtained gives the desired upper bound. 

\subsection{Upper Bound on the Persistent-Policy Cost}
First, we prove that the cost of the persistent policy is bounded by $\bar J_{\bm\nu}^{\rm pers}(\bm m) \leq (1-\gamma)\bm\nu^\top \bigl(I-\bm\Gamma(\bm m)\bigr)^{-1} \widetilde{\bm C}_K(\bm m)$. For a cutoff $K\geq1$, define the tail factor
\begin{equation}
R_K :=
\frac{\lambda+\eta(K+1)}{1-\beta} + \frac{\eta\beta}{(1-\beta)^2}.
\label{eq:persistent_tail_factor}
\end{equation}
We decompose the policy cost into regeneration cycles. Starting from a cycle whose latest successful observation is $i$, the discounted weight with which the next cycle starts from state $k$ is
\begin{align}
\Gamma_{ik}(\bm m)
&=\sum_{n=m_i}^{\infty} \gamma^n s\eta^{n-m_i}(\bm P^n)_{ik}\nonumber\\[-1mm]
&=s\gamma^{m_i} e_i\bm P^{m_i}(I-\beta\bm P)^{-1}e_k^\top,
\end{align}
which agrees with \eqref{eq:persistent_regeneration_matrix}. The inverse exists because $\beta=\gamma\eta<1$. In addition,
\[
\sum_{k=1}^N\Gamma_{ik}(\bm m) =
\frac{\gamma^{m_i}s}{1-\gamma\eta} \leq\frac{\gamma s}{1-\gamma\eta}<1.
\]
Thus $\|\bm\Gamma(\bm m)\|_\infty<1$, so $I-\bm\Gamma(\bm m)$ is invertible and
\[
(I-\bm\Gamma(\bm m))^{-1} =
\sum_{r=0}^{\infty}(\bm\Gamma(\bm m))^r
\]
is nonnegative.

The exact cost accumulated before the next success is
\begin{align}
C_i(\bm m):={}&
\sum_{n=1}^{m_i-1}\gamma^{n-1}g_i(n)\nonumber\\[-1mm]
&+\sum_{n=m_i}^{\infty} \gamma^{n-1}\eta^{n-m_i} \bigl(\lambda+\eta g_i(n)\bigr), \quad i\in\mathcal X.
\label{eq:persistent_cycle_cost_exact}
\end{align}
The first sum is the idle portion of the cycle. In the second sum, all attempts before time slot $n$ have failed; the sampling cost is paid in time slot $n$, while the AoII cost is paid only if that attempt also fails.

To obtain a computable expression, split the second sum at $K$ and denote its tail by
\[
T_{i,K} :=
\sum_{n=K+1}^{\infty} \gamma^{n-1}\eta^{n-m_i} \bigl(\lambda+\eta g_i(n)\bigr).
\]
Since $g_i(n)\leq n$ and $\beta=\gamma\eta$, this tail satisfies
\begin{align*}
T_{i,K} &\leq\gamma^{m_i-1}\beta^{K+1-m_i} \sum_{r=0}^{\infty}\beta^r \bigl(\lambda+\eta(K+1+r)\bigr)\\&
= \gamma^{m_i-1}\beta^{K+1-m_i}R_K.
\end{align*}
Combining the finite part of the cycle cost with the previous tail estimate gives $\widetilde C_{i,K}(\bm m)$ as defined in \eqref{eq:persistent_cycle_cost_cutoff}, for which $C_i(\bm m)\leq\widetilde C_{i,K}(\bm m)$. Let $\bm C(\bm m)$ and $\widetilde{\bm C}_K(\bm m)$ collect the exact and upper cycle costs, respectively.

The regenerative value equation is
\begin{equation}
\bm V(\bm m) = \bm C(\bm m)+\bm\Gamma(\bm m)\bm V(\bm m).
\label{eq:persistent_linear_system_exact}
\end{equation}
Consequently,
\[
\bm V(\bm m) =
(I-\bm\Gamma(\bm m))^{-1}\bm C(\bm m) \leq (I-\bm\Gamma(\bm m))^{-1}\widetilde{\bm C}_K(\bm m).
\]
Multiplying by $\bm\nu^\top$ and by $1-\gamma$ gives the desired upper bound.\qed

\subsection{Lower Bound on the Optimal-Policy Cost}
Next, we prove that the cost of the optimal policy is bounded by $\bar J_{\bm\nu}^{*}
\geq (1-\gamma)\sum_{i=1}^{N}\nu_i \min_{a\in\mathcal A_T} \left\{C_{i,T}(a)+\Gamma_T(a)v_T\right\}$. We need two intermediate results:

\subsubsection{Regenerative relaxation} For every finite horizon $T\geq1$ and every state $i\in\mathcal X$, the optimal value satisfies
\begin{equation}
V_i^*\geq
\inf_{u\in[0,1]^T}
\left\{C_{i,T}(u) + \Gamma_T(u)V_{\min}^*\right\}.
\label{eq:finite_regenerative_relaxation_bound}
\end{equation}

To prove this, fix an admissible policy and suppose that the latest successful transmission observed state $i$. We examine the process until the next successful transmission or until time slot $T$, whichever occurs first.

Let $u_n\in[0,1]$ be the conditional probability that the policy requests a transmission in time slot $n$, given that no success has occurred in the preceding time slots. For $u=(u_1,\dots,u_T)$, the recursion in \eqref{eq:no_regeneration_probability} equivalently gives
\[
Q_n(u)=\prod_{r=1}^{n}(1-su_r), \qquad n=1,\dots,T.
\]
Since the reduced state is $(i,n)$ whenever the cycle reaches time slot $n$, the discounted transmission and AoII costs accumulated up to the cutoff are $C_{i,T}(u)$ as defined in \eqref{eq:finite_relaxed_block_cost}. Indeed, the sampling cost in time slot $n$ is incurred on the event that the cycle has survived through time slot $n-1$, while the AoII cost $g_i(n)$ is incurred only if the cycle also survives time slot $n$.

If a successful transmission occurs in time slot $n\leq T$, the future optimal value is at least $V_{\min}^*$. On the event of no success by the cutoff, we relax the problem by granting a free regeneration at time slot $T$. This can only reduce the continuation cost. The corresponding discounted regeneration weight is $\Gamma_T(u)$ as defined in \eqref{eq:finite_regeneration_prob}. It satisfies $0\leq\Gamma_T(u)\leq\gamma<1$. Therefore every admissible policy induces a vector $u\in[0,1]^T$ such that
\[
V_i^*\geq C_{i,T}(u)+\Gamma_T(u)V_{\min}^*.
\]
To conclude the proof, take the infimum over $u$.\qed

\subsubsection{Vertex reduction} Let $\mathcal A_T:=\{0,1\}^T$. For every $i\in\mathcal X$, every $z\geq0$, and every $T\geq1$,
\begin{equation}
\label{eq:vertex_reduction}
\inf_{u\in[0,1]^T} \left\{C_{i,T}(u) + \Gamma_T(u)z\right\} =
\min_{a\in\mathcal A_T}\left\{C_{i,T}(a)+\Gamma_T(a)z\right\}
\end{equation}
and
\begin{equation}
\label{eq:min_inf}
\inf_{u\in[0,1]^T} \frac{C_{i,T}(u)}{1-\Gamma_T(u)} =
\min_{a\in\mathcal A_T}\frac{C_{i,T}(a)}{1-\Gamma_T(a)}.
\end{equation}

To prove this, fix all components of $u$ except $u_r$. Since
\[
Q_n(u)=\prod_{h=1}^{n}(1-su_h),
\]
each $Q_n(u)$ is affine in $u_r$ when the other components are fixed. Hence both $C_{i,T}(u)$ and $\Gamma_T(u)$ are affine in $u_r$, and so is $C_{i,T}(u)+\Gamma_T(u)z$. Its minimum over $u_r\in[0,1]$ is therefore reached at an endpoint. Applying this argument successively to all coordinates shows that the minimum over $[0,1]^T$ is reached at a vertex of the hypercube, namely at some $a\in\mathcal A_T$.

For the ratio, the numerator and denominator are affine in $u_r$, and the denominator is strictly positive because $\Gamma_T(u)\leq\gamma<1$. The ratio is therefore linear in $u_r$ and has a derivative of constant sign wherever defined. Its minimum is again reached at an endpoint. Repeating the coordinatewise argument proves the second claim.\qed

We can now prove the desired lower bound. For each post-regeneration state $i$, define
\[
w_{i,T} :=
\min_{a\in\mathcal A_T} \left\{C_{i,T}(a)+\Gamma_T(a)v_T\right\},
\]
where $v_T$ is defined in \eqref{eq:v_regenerative}, and set
\begin{equation}
L_{\bm\nu,T} := (1-\gamma)\sum_{i=1}^{N}\nu_iw_{i,T}.
\label{eq:regenerative_lower_bound_def}
\end{equation}

Let $V_{\min}^*:=\min_{i\in\mathcal X}V_i^*$. By \eqref{eq:finite_regenerative_relaxation_bound} and \eqref{eq:vertex_reduction},
\[
V_i^*\geq
\min_{a\in\mathcal A_T} \left\{C_{i,T}(a)+\Gamma_T(a)V_{\min}^*\right\}
\]
for every $i$. Choose $i_0$ such that $V_{i_0}^*=V_{\min}^*$ and let $a^*$ reach the minimum for that state. Then
\[
V_{\min}^*\geq
C_{i_0,T}(a^*)+\Gamma_T(a^*)V_{\min}^*,
\]
which, since $\Gamma_T(a^*)<1$, implies
\[
V_{\min}^*\geq
\frac{C_{i_0,T}(a^*)}{1-\Gamma_T(a^*)} \geq v_T.
\]
Using $V_{\min}^*\geq v_T$ and $\Gamma_T(a)\geq0$ in the relaxation gives $V_i^*\geq w_{i,T}$. Therefore,
\[
\bar J_{\bm\nu}^{*} =
(1-\gamma)\sum_{i=1}^{N}\nu_iV_i^*\geq
(1-\gamma)\sum_{i=1}^{N}\nu_iw_{i,T} =
L_{\bm\nu,T}.
\]\qed

\subsection{Suboptimality Performance Gap}
We now combine the two sides of the bound. The quantity $U_{\bm\nu,K}(\bm m)$ is an upper bound on the normalized discounted cost of the persistent policy, and $L_{\bm\nu,T}$ is a lower bound on the normalized discounted cost of the optimal policy. Therefore, their difference immediately gives our upper bound on the performance gap between the two policies.\qed

\section{Proof of Lemma \ref{lem:map_stabilization}}
\label{app:map_stabilization}
Since $\bm P$ is ergodic by assumption,
\[
e_i\bm P^n\longrightarrow \bm\omega, \qquad i\in\mathcal X.
\]
Let
\[
\delta_\omega :=
\omega_{x_\omega} - \max_{j\neq x_\omega}\omega_j.
\]
Since $x_\omega$ is the unique mode of $\bm\omega$, we have $\delta_\omega>0$.

Because the state space $\mathcal X$ is finite, there exists a finite $n_0$ such that, for every $i\in\mathcal X$ and every $n\geq n_0$,
\[
\left\|e_i\bm P^n-\bm\omega\right\|_\infty < \frac{\delta_\omega}{3}.
\]
Therefore, for every $j\neq x_\omega$,
\begin{align*}
(e_i\bm P^n)_{x_\omega} - (e_i\bm P^n)_j & >
\omega_{x_\omega} - \omega_j - \frac{2\delta_\omega}{3}\\&
\geq \frac{\delta_\omega}{3} > 0.
\end{align*}
Thus, $x_\omega$ is the unique MAP estimate for every initial state $i$ and every $n\geq n_0$, and the set in the definition of $\kappa_{\rm MAP}$ is nonempty and $\kappa_{\rm MAP}\leq n_0-1<\infty$.\qed

\section{Proof of Lemma \ref{lem:hybrid_active}}
\label{app:hybrid_active}

We first show that $U^{\rm AT}$ uniformly bounds $V_i^E(1)$ for both estimators. Let $M^{{\rm AT},E}$ be the maximum cost of the always-transmit policy over the initial states $(i,1)$. Since the optimal value cannot exceed the cost of this policy,
\[
\max_iV_i^E(1)\leq M^{{\rm AT},E}.
\]
Using $g_i^E(n)\leq n$, the cost accumulated before the next success and the discounted regeneration weight are bounded by
\[
C^{\rm AT} =
\frac{\lambda}{1-\beta} + \frac{\eta}{(1-\beta)^2}, \qquad \Gamma^{\rm AT}=\frac{\gamma s}{1-\beta}.
\]
Therefore,
\[
M^{{\rm AT},E} \leq C^{\rm AT}+\Gamma^{\rm AT}M^{{\rm AT},E},
\]
and hence
\[
\max_iV_i^E(1)\leq
\frac{C^{\rm AT}}{1-\Gamma^{\rm AT}} = U^{\rm AT}.
\]

Suppose now that the passive action is optimal at $(i,n)$. Comparing the passive and active branches of \eqref{eq:hybrid_reduced_bellman} yields
\[
\begin{aligned}
g_i^E(n)+\gamma V_i^E(n+1)\leq{}&
\lambda+\eta g_i^E(n) + \gamma\eta V_i^E(n+1)\\[-1mm]&
+ \gamma s\sum_{k=1}^{N}(\bm P^n)_{ik}V_k^E(1).
\end{aligned}
\]
Since $s>0$, this implies
\[
g_i^E(n)+\gamma V_i^E(n+1) \leq
\frac{\lambda}{s} + \gamma\sum_{k=1}^{N}(\bm P^n)_{ik}V_k^E(1) \leq
\frac{\lambda}{s}+\gamma U^{\rm AT}.
\]
The value function is nonnegative, so a necessary condition for passivity to be optimal is
\[
g_i^E(n)\leq\frac{\lambda}{s}+\gamma U^{\rm AT}.
\]
Condition \eqref{eq:hybrid_forced_active_condition} violates this inequality; therefore the active action must be optimal.\qed

\section{Proof of Theorem \ref{thm:hybrid_map_bound}}
\label{app:hybrid_map_bound}
Fix $R\geq\kappa$. We proceed in three steps.

\subsubsection{Regenerative discrepancy bound under a common policy}
Fix $E\in\{{\rm map},{\rm hyb}\}$ and a deterministic stationary policy $\phi$ whose first $R$ actions in each cycle are represented by a table $\bm u\in\mathcal U_R^E$, that is,
\[
u_i(n)=\phi(i,n),\qquad i\in\mathcal X,\quad n=1,\ldots,R.
\]
In particular, by Lemma~\ref{lem:hybrid_active}, this property is satisfied by an optimal Bellman selector for estimator $E$.

Apply the same policy $\phi$ to the MAP and hybrid models and couple their processes by using the same transition outcomes. This coupling is possible because the estimator changes only the AoII cost and does not affect the transition law. Consequently, the two models visit the same states and every successful transmission occurs in the same slot and reveals the same state.

Starting from the post-regeneration state $(i,1)$, define the total discounted absolute stage-cost discrepancy under $\phi$ as
\begin{equation}
z_i^{\phi} :=
\mathbb E_i^{\phi}\!\left[ \sum_{t=0}^{\infty}\gamma^t
 (1-sA_t)\varepsilon_{I_t}^{\kappa}(N_t) \right],
\label{eq:hybrid_common_policy_discrepancy}
\end{equation}
where $(I_t,N_t)$ is the common reduced-state process and $A_t=\phi(I_t,N_t)$. Since the transmission cost is identical in the two models, the triangle inequality gives
\begin{equation}
\left|J_i^{{\rm map},\phi}(1)-J_i^{{\rm hyb},\phi}(1)\right| \leq z_i^{\phi}.
\label{eq:hybrid_common_policy_value_difference}
\end{equation}
Moreover, the two estimators coincide during the first $\kappa$ slots of every cycle, and hence
\[
\varepsilon_i^{\kappa}(n)=0,\qquad n\leq\kappa.
\]

We now decompose $z_i^{\phi}$ at the first successful transmission in the current cycle. The probability that the cycle continues until time slot $n$ is $Q_i^{\bm u}(n)$. Therefore, the expected discounted discrepancy accumulated during the first $R$ slots is
\[
\sum_{n=1}^{R}\gamma^{n-1}Q_i^{\bm u}(n) \varepsilon_i^{\kappa}(n).
\]
If no successful transmission occurs by the end of slot $R$, then at slot $(R+1+r)$ after the beginning of the current cycle the elapsed no-success duration is at most $R+1+r$. Since $0\leq g_j^{\rm map}(\ell),g_j^{\rm hyb}(\ell)\leq\ell$, the absolute discrepancy at that slot is at most $R+1+r$. Thus, on the event that the cycle survives through slot $R$, the entire remaining discounted discrepancy, including all possible later regeneration cycles, is bounded by
\[
\sum_{r=0}^{\infty}\gamma^r(R+1+r) =
\frac{R+1}{1-\gamma}+\frac{\gamma}{(1-\gamma)^2}.
\]
This gives exactly the tail term in $[\bm c_{\kappa,R}(\bm u)]_i$.

On the other hand, if the first successful transmission occurs in slot $n\leq R$ and reveals state $k$, the two models regenerate simultaneously at $(k,1)$, and the remaining discrepancy is $z_k^{\phi}$. The corresponding discounted transition weight is
\[
\gamma^n s u_i(n)Q_i^{\bm u}(n-1)(\bm P^n)_{ik}.
\]
Summing over all $n\leq R$ yields $[\bm\Theta_R(\bm u)]_{ik}$. Hence, with $\bm z^{\phi}:=(z_1^{\phi},\ldots,z_N^{\phi})^\top$, the decomposition gives the regenerative inequality
\begin{equation}
\bm z^{\phi} \leq \bm c_{\kappa,R}(\bm u)
+ \bm\Theta_R(\bm u)\bm z^{\phi}.
\label{eq:hybrid_regenerative_discrepancy_inequality}
\end{equation}

For every $i$,
\begin{align*}
\sum_{k=1}^{N}[\bm\Theta_R(\bm u)]_{ik}
&=
\sum_{n=1}^{R}\gamma^n s u_i(n)Q_i^{\bm u}(n-1)\\
&\leq
\gamma\sum_{n=1}^{R}s u_i(n)Q_i^{\bm u}(n-1)
\leq\gamma<1.
\end{align*}
Therefore, $\|\bm\Theta_R(\bm u)\|_\infty\leq\gamma<1$, so $I-\bm\Theta_R(\bm u)$ is invertible and
\[
(I-\bm\Theta_R(\bm u))^{-1}
=
\sum_{r=0}^{\infty}\bm\Theta_R(\bm u)^r
\]
is nonnegative. From \eqref{eq:hybrid_regenerative_discrepancy_inequality},
\begin{equation}
\bm z^{\phi}
\leq
(I-\bm\Theta_R(\bm u))^{-1}
\bm c_{\kappa,R}(\bm u).
\label{eq:hybrid_common_policy_certificate}
\end{equation}
Combining \eqref{eq:hybrid_common_policy_value_difference} and
\eqref{eq:hybrid_common_policy_certificate} gives a discrepancy certificate for any policy that is in common between the two estimators.

\subsubsection{Comparison of the two optimal post-regeneration values}
Let $\phi_{\rm hyb}^*$ be an optimal deterministic policy for the hybrid model, and let $\bm u_{\rm hyb}^*$ be its first-$R$ action table. By Lemma~\ref{lem:hybrid_active}, $\bm u_{\rm hyb}^*\in\mathcal U_R^{\rm hyb}$. Applying the previous bound to the common policy $\phi_{\rm hyb}^*$ and then maximizing over $\mathcal U_R^{\rm hyb}$ gives
\[
\begin{aligned}
V_i^{\rm map}(1)
&\leq J_i^{{\rm map},\phi_{\rm hyb}^*}(1)\\
&\leq J_i^{{\rm hyb},\phi_{\rm hyb}^*}(1)
+B_{\kappa,R}^{\rm hyb}(i)\\
&=V_i^{\rm hyb}(1)+B_{\kappa,R}^{\rm hyb}(i).
\end{aligned}
\]
Similarly, using an optimal MAP policy $\phi_{\rm map}^*$ as the common policy in the two models yields
\[
V_i^{\rm hyb}(1)
\leq V_i^{\rm map}(1)+B_{\kappa,R}^{\rm map}(i).
\]
Therefore,
\begin{equation}
\left|V_i^{\rm map}(1)-V_i^{\rm hyb}(1)\right|
\leq
\max\left\{B_{\kappa,R}^{\rm map}(i),
B_{\kappa,R}^{\rm hyb}(i)\right\}.
\label{eq:hybrid_post_regeneration_bound_proof}
\end{equation}

\subsubsection{Synchronized initial belief}
At $b_0$, the passive action has zero cost and leads to $(i_0,1)$ after the source transition, while the active action provides the same information and additionally pays $\lambda$. Hence, the passive action must be optimal for both estimators and
\[
V_{\rm map}^*(b_0)=\gamma V_{i_0}^{\rm map}(1),
\qquad
V_{{\rm hyb},\kappa}^*(b_0)=\gamma V_{i_0}^{\rm hyb}(1).
\]
To conclude the proof, multiply \eqref{eq:hybrid_post_regeneration_bound_proof} by $\gamma$, set $i=i_0$, and take the infimum over $R\geq\kappa$.\qed

\section{Proof of Theorem \ref{thm:indexability}}
\label{app:indexability}
We omit the source index $j$ to simplify the notation. Fix $W'>W$. Since only the passive action depends explicitly on the subsidy, increasing the subsidy can reduce the one-stage cost by at most $W'-W$. Therefore, for every reduced state $(i,n)$,
\begin{equation}
0\leq V_i(n,W)-V_i(n,W')
\leq\frac{W'-W}{1-\gamma}.
\label{eq:value_subsidy_lipschitz}
\end{equation}
Using \eqref{eq:relaxed_single_arm_Q0}--\eqref{eq:relaxed_single_arm_Q1},
\begin{align}
\Psi_i(n,W)
={}&-s g_i(n)+W\nonumber\\[-1mm]
&+\gamma s\left(
\sum_{k=1}^{N}(\bm P^n)_{ik}V_k(1,W)
-V_i(n+1,W)
\right).
\label{eq:whittle_gap_expanded}
\end{align}
For every state, define
\[
\zeta_i(n):=V_i(n,W)-V_i(n,W').
\]
By \eqref{eq:value_subsidy_lipschitz},
\[
0\leq\zeta_i(n)\leq\frac{W'-W}{1-\gamma}.
\]
It follows that
\begin{align}
\Psi_i(n,W') - \Psi_i(n,W) ={}&
(W'-W)+\gamma s\zeta_i(n+1)\nonumber\\[-1mm]&
-\gamma s\sum_{k=1}^{N}(\bm P^n)_{ik}\zeta_k(1)\nonumber\\
\geq{}&(W'-W)\left(1-\frac{\gamma s}{1-\gamma}\right).
\label{eq:whittle_gap_monotonicity_bound}
\end{align}
Under \eqref{eq:multi_source_indexability_condition}, the last expression is nonnegative. Hence, $\Psi_i(n,W')\geq\Psi_i(n,W)$ for every state. If $(i,n)$ is passive at subsidy $W$, it remains passive at every larger subsidy, and the arm is indexable.\qed

\end{document}